\documentclass[UKenglish,cleveref,numberwithinsect,thm-restate,a4paper,11pt]{article}
\usepackage[utf8]{inputenc}
\usepackage{a4wide}
\usepackage{hyperref}
\hypersetup{colorlinks=true,citecolor=blue,linkcolor=blue,urlcolor=blue}
\usepackage{amsmath}
\usepackage{amssymb}
\usepackage{amsthm}
\usepackage{thmtools, thm-restate}
\usepackage{enumerate}
\usepackage{pdfpages}
\usepackage[noadjust]{cite}

\usepackage[UKenglish]{babel}
\usepackage[UKenglish]{isodate}

\usepackage{comment}
\usepackage{graphicx}
\usepackage{blkarray}
\usepackage{tikz, tikz-cd}
\usetikzlibrary{matrix}
\usetikzlibrary{shapes}
\usetikzlibrary{arrows,decorations.markings, tikzmark}
\usetikzlibrary{decorations.pathreplacing,calligraphy,backgrounds}
\usetikzlibrary{arrows.meta,arrows}
\usepackage{xcolor}
\usepackage{mathtools}
\usepackage{stmaryrd}
\usepackage{tabularx}
\usepackage{booktabs}
\usepackage{multirow}
\usepackage{algorithm}
\usepackage{algpseudocode}
\usepackage{svg}

\theoremstyle{plain}
\newtheorem{theorem}{Theorem}
\newtheorem*{theorem*}{Theorem}
\newtheorem{corollary}[theorem]{Corollary}
\newtheorem{lemma}[theorem]{Lemma}
\newtheorem{observation}[theorem]{Observation}

\newtheorem{conj}[theorem]{Conjecture}

\theoremstyle{definition}
\newtheorem{definition}[theorem]{Definition}
\newtheorem{remark}[theorem]{Remark}
\newtheorem{example}[theorem]{Example}

\newcommand{\homsprob}{\textsc{CQ}}
\newcommand{\Homs}[2]{\mathsf{Hom}(#1 \to #2)}
\newcommand{\homs}[2]{\mathsf{hom}(#1 \to #2)}

\newcommand{\conjunctfull}[1]{\boldsymbol{\wedge}\left(#1\right)} 
\newcommand{\conjunctclass}[1]{\conjunctfull{#1}}
\newcommand{\conjuncts}[2]{\conjunctfull{#1\vert_{#2}}}
\newcommand{\cequiv}{\sim}

\newcommand{\contract}[1]{\mathsf{contract}(#1)}
\newcommand{\cequivalent}{\#\text{equivalent}}
\newcommand{\cminimal}{\#\text{minimal}}
\newcommand{\ccore}{\#\text{core}}
\newcommand{\ccores}{\#\text{cores}}

\newcommand{\tw}{\mathsf{tw}}

\newcommand{\ccStyle}[1]{\mathrm{#1}}

\newcommand{\ccW}[1]{\ccStyle{W[#1]}}
\newcommand{\ccSharpW}[1]{\#\ccStyle{W[#1]}}

\newcommand{\Ans}[2]{\mathsf{Ans}(#1 \to #2)}
\newcommand{\ans}[2]{\mathsf{ans}(#1 \to #2)}

\newcommand{\fptred}{\leq^{\mathsf{FPT}}}

\newcommand{\ucq}{\textsc{UCQ}}
\newcommand{\meta}{\textsc{Meta}}

\newcommand{\calA}{\mathcal{A}}
\newcommand{\calB}{\mathcal{B}}
\newcommand{\calD}{\mathcal{D}}

\newcommand{\abs}[1]{\left\vert #1 \right\vert}

\title{Counting Answers to Unions of Conjunctive Queries: Natural Tractability Criteria and Meta-Complexity\thanks{For the purpose of Open Access, the
authors have applied a CC BY public copyright licence to any Author Accepted Manuscript version arising
from this submission. All data is provided in full in the results section of this paper. 
A previous version of this paper appeared in the proceedings of the ACM on Management of Data (PODS 2024).
Stanislav \v{Z}ivn\'{y} was supported by UKRI EP/X024431/1.}
}

\author{
Jacob Focke \\ \small{CISPA Helmholtz Center for Information Security} \\ \small{Saarbrücken}\\ \small{Germany}
\and Leslie Ann Goldberg \\ \small{Department of Computer Science} \\ \small{University of Oxford}\\ \small{United Kingdom} \and Marc Roth \\ \small{School of Electronic Engineering and Computer Science}\\ \small{Queen Mary University of London}\\ \small{United Kingdom}  \and 
Stanislav {\v{Z}}ivn{\'y} \\ \small{Department of Computer Science} \\ \small{University of Oxford}\\ \small{United Kingdom}
}

\date{}

\begin{document}

\maketitle

\begin{abstract}
We study the problem of counting answers to unions of conjunctive queries (UCQs) under structural restrictions on the input query. Concretely, given a class $C$ of UCQs, the problem $\#\textsc{UCQ}(C)$ 
provides as input a UCQ $\Psi \in C$ and a database $\mathcal{D}$ and the problem is to compute the number of answers of $\Psi$ in $\mathcal{D}$.
    
Chen and Mengel [PODS'16] have shown that for any recursively enumerable class~$C$, the problem $\#\textsc{UCQ}(C)$ is either fixed-parameter tractable or hard for 
one of the parameterised complexity classes $\mathrm{W}[1]$ or $\#\mathrm{W}[1]$. However, their tractability criterion is unwieldy in the sense that, given any concrete class $C$ of UCQs, 
it is not easy to determine how hard it is to count answers to queries in $C$. Moreover, given a 
single specific UCQ $\Psi$,  
it is not easy to determine how hard it is to count answers to~$\Psi$.

In this work, we address the question of finding a \emph{natural} tractability criterion: The combined conjunctive query of a UCQ $\Psi=\varphi_1 \vee \dots \vee \varphi_\ell$ is the conjunctive query $\conjunctfull{\Psi}= \varphi_1 \wedge \dots \wedge \varphi_\ell$. We show that under natural closure properties of $C$, the problem $\#\textsc{UCQ}(C)$ is fixed-parameter tractable if and only if the combined conjunctive queries of UCQs in~$C$, and their contracts, have bounded treewidth. A contract of a conjunctive query is an augmented structure, taking into account how the quantified variables are connected to the free variables
--- if all variables are free, then a conjunctive query is equal to its contract; in this special case the criterion for fixed-parameter tractability of $\#\textsc{UCQ}(C)$ thus simplifies to the combined queries having bounded treewidth.

Finally, we give evidence that a closure property on $C$ is necessary for obtaining a natural tractability criterion: We show that even for a single UCQ $\Psi$, the meta problem of deciding whether $\#\textsc{UCQ}(\{\Psi\})$ can be solved in time $O(|\mathcal{D}|^d)$ is  $\mathrm{NP}$-hard for any fixed $d\geq 1$. Moreover, we prove that a known exponential-time algorithm for solving the meta problem is optimal under assumptions from fine-grained complexity theory. As a corollary of our reduction, we also establish that approximating the Weisfeiler-Leman-Dimension of a UCQ is $\mathrm{NP}$-hard.
\end{abstract}

\section{Introduction}\label{sec:intro}

Conjunctive queries are among the most fundamental and well-studied objects in database theory~\cite{ChandraM77,Yannakakis81,AbiteboulHV95,KolaitisV95,KolaitisV00,Vardi00}. A conjunctive query (CQ) $\varphi$ with free variables $X=\{x_1, \ldots, x_k\}$ and quantified variables $Y=\{y_1, \ldots, y_d\}$ is of the form
\[
\varphi(X) =\exists Y\, R_1(\mathbf{t}_1) \wedge \ldots \wedge R_n(\mathbf{t}_n),
\]
where $R_1, \ldots, R_n$ are relational symbols and each $\mathbf{t}_i$ is a tuple of variables from $X\cup Y$. A database $\calD$ consists of a set of elements $U(\calD)$, denoted the \emph{universe} of $\calD$, and a set of relations over this universe. The corresponding relation symbols are the \emph{signature} of $\calD$. If $R_1, \dots, R_n$ are in the signature of $\calD$ then an \emph{answer} of $\varphi$ in $\calD$ is an assignment $a\colon X\to U(\calD)$ that has an extension to the existentially quantified variables $Y$ that agrees with all the relations $R_1,\ldots, R_n$. Even more expressive is a union of conjunctive queries (UCQ). Such a union is of the form
\[
\Psi(X) = \varphi_1(X) \vee \ldots \vee \varphi_\ell(X),
\]
where each $\varphi_i(X)$ is a CQ with free variables $X$. An answer to $\Psi$ is then any assignment that is answer to at least one of the CQs in the union. 

Since evaluating a given CQ on a given database is NP-complete~\cite{ChandraM77} a lot of research focused on finding tractable classes of CQs. 
A fundamental result by Grohe, Schwentick, and Segoufin~\cite{GroheSS01} 
established that the tractability of evaluating all CQs of bounded arity whose Gaifman graph is in some class of graphs $C$ depends on whether or not the treewidth in~$C$ is bounded.

More generally, finding an answer to a conjunctive query can be cast as finding a (partial) homomorphism between relational structures, and therefore is closely related to the framework of constraint satisfaction problems. In this setting, Grohe~\cite{Grohe07} showed that treewidth modulo homomorphic equivalence is the right criterion for tractability. There is also an important line of work~\cite{Gottlob02:jcss-hypertree,Grohe14:talg,Marx10:talg} culminating in the fundamental work by Marx~\cite{Marx13:jacm} that investigates the parameterised complexity for classes of queries with unbounded arity. In general, tractability of conjunctive queries is closely related to how ``tree-like'' or close to acyclic they are.

Counting answers to CQs has also received significant attention in the past~\cite{DalmauJ04,PichlerS13,DurandM13,GrecoS14,ChenM16,DellRW19,ArenasNew}. Chen and Mengel~\cite{ChenM15} gave a complete classification for the counting problem on classes of CQs (with bounded arity) in terms of a natural criterion loosely based on treewidth. They present a trichotomy into fixed-parameter tractable, $\ccW{1}$-complete, and $\ccSharpW{1}$-hard cases. In subsequent work~\cite{ChenM16}, this classification was extended to unions of conjunctive queries (and to even more general queries in~\cite{DellRW19}). However, for UCQs, the established criteria for tractability and intractability are implicit (see~\cite[Theorems 3.1 and 3.2]{ChenM16}) in the sense that, given a specific UCQ $\Psi$, it is not at all clear how hard it is to count answers to $\Psi$ based on the criteria in~\cite{ChenM16}. To make this more precise: It is not even clear whether we can, in polynomial time in the size of $\Psi$, determine whether answers to $\Psi$ can be counted in linear time in the input database.

\subsection{Our contributions}
With the goal of establishing a more practical tractability criterion for counting answers to UCQs, we explore the following two main questions in this work:

\begin{quote}
Q1) Is there a natural criterion that captures the fixed-parameter tractability of counting answers to a class of UCQs, parameterised by the size of the query?
\end{quote}
\begin{quote}
Q2) Is there a natural criterion that captures whether counting answers to a single fixed UCQ is linear-time solvable (in the size of a given database)?
\end{quote}

\paragraph{Question Q1): Fixed-Parameter Tractability.}
For a class $C$ of UCQs, we consider the problem $\#\ucq(C)$ that takes as input a UCQ $\Psi$ from $C$ and a database $\calD$, and asks for the number $\ans{\Psi}{\calD}$ of answers of $\Psi$ in $\calD$.
We assume that the arity of the UCQs in $C$ is bounded, that is, there is constant $c$ such that each relation that appears in some query in $C$ has arity at most $c$. 
As explained earlier, due to a result of Chen and Mengel~\cite{ChenM16}, there is a known but rather unwieldy tractability criterion for $\#\ucq(C)$, when the problem is parameterised by the size of the query. On a high level, the number of answers of a UCQ $\Psi$ in a given database can be expressed as a finite linear combination of CQ answer counts, using the principle of inclusion-exclusion. This means that
\[
    \ans{\Psi}{\calD} = \sum_i c_i \cdot \ans{\varphi_i}{\calD},
\] 
where each $\varphi_i$ is simply a conjunctive query (and not a union thereof). We refer to this linear combination as the \emph{CQ expansion} of $\Psi$.
Chen and Mengel showed that the parameterised complexity of computing $\ans{\Psi}{\calD}$ is guided by the hardest term in the respective CQ expansion. The complexity of computing these terms is simply the complexity of counting the answers of a conjunctive query, and this is well understood~\cite{ChenM15}. Hence, the main challenge for this approach is to understand the linear combination, i.e., to understand for which CQs the corresponding coefficients are non-zero.
The problem is that the coefficients $c_i$ of these linear combinations are alternating sums, which in similar settings have been observed to encode algebraic and even topological invariants~\cite{RothS18}. This makes it highly non-trivial to determine which CQs actually contribute to the linear combination.
We introduce the concepts required to state this classification informally, the corresponding definitions are given in Section~\ref{sec:preliminaries}.
 
We first give more details about the result of~\cite{ChenM15}.
Let $\Gamma(C)$ be the class of those conjunctive queries that contribute to the CQ expansion of at least one UCQ in $C$, and that additionally are what we call $\cminimal$. Intuitively, a conjunctive query $\varphi$ is $\cminimal$ if there is no proper subquery $\varphi'$ of $\varphi$ that has the same number of answers as $\varphi$ in every given database. 
Then the tractability criterion depends on the treewidth of the CQs in $\Gamma(C)$. It also depends on the treewidth of the corresponding class $\contract{\Gamma(C)}$ of contracts (formally defined in Definition~\ref{def:contract}), which is an upper bound of what is called the ``star size'' in~\cite{DurandM13} and the ``dominating star size'' in~\cite{DellRW19}.
Here is the formal statement of the known dichotomy for $\#\ucq(C)$.

\begin{restatable}[\cite{ChenM16}]{theorem}{ChenMUCQ}
    \label{thm:implicit_UCQ_dicho}
Let $C$ be a recursively enumerable class of UCQs of bounded arity. If the treewidth of $\Gamma(C)$ and of $\contract{\Gamma(C)}$ is bounded, then $\#\ucq(C)$ is fixed-parameter tractable. Otherwise, $\#\ucq(C)$ is $\ccW{1}$-hard.
\end{restatable}

We investigate under which conditions this dichotomy can be simplified. We show that for large classes of UCQs there is actually a much more natural tractability criterion that does not rely on $\Gamma(C)$, i.e., here the computation of the coefficients of the linear combinations as well as the concept of \#minimality do not play a role.
We first show a simpler classification for UCQs without existential quantifiers. 
To state the results we require some additional definitions: The \emph{combined query} $\conjunctclass{\Psi}$ of a UCQ $\Psi(X)=\varphi_1(X) \vee \dots \vee \varphi_\ell(X)$ is the conjunctive query obtained from $\Psi$ by replacing each disjunction by a conjunction, that is
\[\conjunctclass{\Psi} = \varphi_1(X) \wedge \dots \wedge \varphi_\ell(X)\,.\]
Given a class of UCQs $C$, we set $\conjunctclass{C}= \{\conjunctclass{\Psi}\mid \Psi \in C \}$. 

It will turn out that the structure of the class of combined queries $\conjunctclass{C}$ determines the complexity of counting answers to UCQs in $C$, given that $C$ has the following natural closure property: We say that $C$ is \emph{closed under deletions} if, for all $\Psi(X)=\varphi_1(X) \vee \dots \vee \varphi_\ell(X)$ and for every $J\subseteq [\ell]$, the subquery $\bigvee_{j\in J}\varphi_j(X)$ is also contained in $C$. For example, any class of UCQs defined solely by the conjunctive queries admissible in the unions (such as unions of acyclic conjunctive queries) is closed under deletions. The following classification resolves the complexity of counting answers to UCQs in classes that are closed under deletions; we will see later that the closedness condition is necessary. Moreover, the tractability criterion depends solely on the structure of the combined query, and not on the terms in the CQ expansion, thus yielding, as desired, a much more concise and natural characterisation. As mentioned earlier, we first state the classification for quantifier-free UCQs.

\begin{restatable}{theorem}{maindichosimple}\label{thm:main_dicho_quantifier_free}
Let $C$ be recursively enumerable class of quantifier-free UCQs of bounded arity. If $\conjunctclass{C}$ has bounded treewidth then $\#\ucq(C)$ is fixed-parameter tractable. 
If $\conjunctclass{C}$ has unbounded treewidth  and $C$ is closed under deletions then $\#\ucq(C)$ is $\ccW{1}$-hard.
\end{restatable}

We emphasise here that Theorem~\ref{thm:main_dicho_quantifier_free} is in terms of the simpler object $\conjunctclass{C}$ instead of the complicated object $\Gamma(C)$.

If we allow UCQs with quantified variables in the class $C$ then the situation becomes more intricate.
Looking for a simple tractability criterion that describes the complexity of $\#\ucq(C)$ solely in terms of $\conjunctclass{C}$ requires some additional effort.
First, for a UCQ $\Psi$ that has quantified variables, $\contract{\Psi}$ is not necessarily the same as $\Psi$,  and therefore the treewidth of the contracts also plays a role. Moreover, the matching lower bound requires some conditions in addition to being closed under deletions. 
Nevertheless, our result is in terms of the simpler objects $\conjunctclass{C}$ and $\contract{\conjunctclass{C}}$ rather than the more complicated $\Gamma(C)$ and $\contract{\Gamma(C)}$.
For Theorem~\ref{thm:main_dicho_quantifiers}, recall that a conjunctive query is \emph{self-join-free} if each relation symbol occurs in at most one atom of the query. 

\begin{restatable}{theorem}{maindichogeneral}\label{thm:main_dicho_quantifiers}
Let $C$ be a recursively enumerable class of UCQs of bounded arity. If $\conjunctclass{C}$ and $\contract{\conjunctclass{C}}$ have bounded treewidth then $\#\ucq(C)$ is fixed-parameter tractable.  Otherwise, if  (I)--(III) are satisfed, then
$\#\ucq(C)$ is $\ccW{1}$-hard. 
\begin{enumerate}
    \item[(I)] $C$ is closed under deletions.
    \item[(II)] The number of existentially quantified variables of queries in $C$ is bounded.
    \item[(III)] The UCQs in $C$ are unions of self-join-free conjunctive queries.
\end{enumerate}
\end{restatable}

In Appendix~\ref{sec:appendix} we show that Theorem~\ref{thm:main_dicho_quantifiers} is tight in the sense that, if any of these conditions is dropped, there are counterexamples to the claim that tractability is guided solely by $\conjunctclass{C}$ and $\contract{\conjunctclass{C}}$.

\paragraph{Question Q2): Linear-Time Solvability.}

Now we turn to the question of linear-time solvability for a single fixed UCQ.
The huge size of databases in modern applications motivates the question of which query problems are actually linear-time solvable. Along these lines, there is a lot of research for enumeration problems~\cite{Yannakakis81,BaganDG07,CarmeliZBCKS22,BeraGLSS22,BerkholzGS20,BerkholzS19}. 

The question whether counting answers  to a conjunctive query $\varphi$ can be achieved in time linear in the given database has been studied previously~\cite{Mengel21}. 
The corresponding dichotomy is well known and was discovered multiple times by different authors in different contexts.\footnote{We remark that~\cite[Theorem 7]{BeraGLSS22} focuses on the special case of graphs and \emph{near} linear time algorithms. However, in the word RAM model with $O(\log n)$ bits, a linear time algorithm is possible~\cite{CarmeliZBCKS22}.}
In these results, the tractability criterion is whether $\varphi$ is \emph{acyclic}, i.e., whether it has a join tree (see~\cite{GottlobGS14}). The corresponding lower bounds are conditioned on a widely used complexity assumption from fine-grained complexity, namely the Triangle Conjecture. We define all of the complexity assumptions that we use in this work in Section~\ref{sec:preliminaries}. There we also formally define the size of a database (as the sum of the size of its signature, its universe, and its relations).

It is well-known that, counting answers to quantifier-free conjunctive queries can be done in linear time if and only if the query is acyclic. The ``only if'' part relies on hardness assumptions from fine-grained complexity theory. Concretely, we have

\begin{theorem}[See Theorem 12 in~\cite{BraultBaron13}, and \cite{BeraGLSS22,Mengel25arxiv}]\label{thm:CQ_dichointro}
Let $\varphi$ be a quantifier-free conjunctive query.
If $\varphi$ is acyclic, then the number of answers of $\varphi$ in a given database $\calD$ can be computed in time linear in the size of $\calD$. Otherwise, if $\varphi$ has arity $2$, then the number of answers cannot be computed in linear time in the size of $\calD$, unless \emph{the Triangle Conjecture} fails.
\end{theorem}
We point out that there is also a conditional lower bound for cyclic quantifier-free conjunctive queries of arity larger than $2$: assuming the Hyperclique Hypothesis, counting answers to such queries can not be done in linear time --- however, since the case of arity $2$ is sufficient for our lower bounds and since this work already includes various lower bound assumptions, we decided not to go into the details of the higher arity case and instead refer the reader to an excellent survey by Mengel~\cite{Mengel25} and~\cite[Theorem 3.8]{Mengel25arxiv}.

We note that the previous theorem is false if quantified variables were allowed as this would require the consideration of  semantic acyclicity\footnote{A conjunctive query is semantically acyclic if and only if its \#core (Definition~\ref{def:core}) is acyclic.} (see~\cite{BarceloRV16}).

Theorem~\ref{thm:CQ_dichointro} yields an efficient way to check whether counting answers to a quantifier-free conjunctive query $\varphi$ can be done in linear time: Just check whether $\varphi$ is acyclic (in polynomial time, see for instance~\cite{GottlobGS14}). 
We investigate the corresponding question for \emph{unions} of conjunctive queries. In stark contrast to Theorem~\ref{thm:CQ_dichointro}, we show that there is no efficiently computable criterion that determines the linear-time tractability of counting answers to \emph{unions} of conjunctive queries, unless some conjectures of fine-grained complexity theory fail.

We first observe that, as in the investigation of question Q1), one can obtain a criterion for linear-time solvability by expressing UCQ answer counts as linear combinations of CQ answer counts. Concretely, by a straightforward extension of previous results, we show that, assuming the Triangle Conjecture, a linear combination of CQ answer counts can be computed in linear time if and only if the answers to each $\cminimal$ CQ in the linear combination can be computed in linear time, that is, if each such CQ is acyclic. However, this criterion is again unwieldy in the sense that, for all we know, it may take time exponential in the size of the respective UCQ to determine whether this criterion holds.

In view of our results for question Q1) about fixed-parameter tractability, one might suspect that a more natural and simpler tractability criterion exists. However, it turns out that even under strong restrictions on the UCQs that we consider, an efficiently computable criterion is unlikely.
We make this formal by studying the following meta problem.\footnote{For the question Q1), considering a similar meta problem is not feasible as, in this case, the meta problem takes as input a class of graphs. If such a class were encoded as a Turing machine, the meta problem would be undecidable by Rice's theorem.}

\vbox{
\begin{description}\setlength{\itemsep}{0pt}
\setlength{\parskip}{0pt}
\setlength{\parsep}{0pt}   			
\item[\bf Name:] $\meta$ 
\item[\bf Input:]   A union $\Psi$ of quantifier-free conjunctive queries. 
\item[\bf Output:]  Is it possible to count answers to $\Psi$ in time linear in the size of any given database~$\calD$. 
\end{description}
}

Restricting the input of $\meta$ to quantifier-free queries is sensible as, without this restriction, the meta problem is known to be $\mathrm{NP}$-hard even for conjunctive queries: If all variables are existentially quantified, then evaluating a conjunctive query can be done in linear time if and only if the query is semantically acyclic~\cite{Yannakakis81} (the ``only if'' relies on standard hardness assumptions). However, verifying whether a conjunctive query is semantically acyclic is already $\mathrm{NP}$-hard~\cite{BarceloRV16}.
In contrast, when restricted to quantifier-free conjunctive queries, the problem $\meta$ is polynomially-time solvable according to Theorem~\ref{thm:CQ_dichointro}.
 
We can now state our main result about the complexity of $\meta$. The hardness results hold under substantial additional input restrictions, which make these results stronger.

\begin{restatable}{theorem}{mainmeta}  
\label{thm:main_meta}
Assume that the Triangle Conjecture is true.
When restricted to input queries of arity at most $2$, $\meta$ can be solved in time $2^{O(\ell)} \cdot |\Psi|^{\mathsf{poly}(\log|\Psi|)}$, where~$\ell$ is the number of conjunctive queries in the union. Moreover, $\meta$ is $\mathrm{NP}$-hard and the following conditional lower bounds apply:
    \begin{itemize}
        \item If ETH is true, then $\meta$ cannot be solved in time $2^{o(\ell)}$.
        \item If the non-uniform ETH is true then
    $\meta \notin \bigcap_{\varepsilon>0} \mathrm{DTime}(2^{\varepsilon \cdot \ell})$.
    \end{itemize}
The lower bounds remain true even if $\Psi$ is a union of self-join-free and acyclic conjunctive queries over a binary signature (that is, of arity~$2$).
\end{restatable}

We make some remarks about Theorem~\ref{thm:main_meta}.
First, it may seem counterintuitive that the algorithmic part of this result
relies on some lower bound conjectures. This is explained by the fact that an algorithmic result for $\meta$ is actually a classification result for the underlying counting problem. The lower bound conjectures are the reason that the algorithm for $\meta$ can answer that a linear-time algorithm is \emph{not} possible for certain UCQs. We also point out that the upper bound remains true for arbitrary arities if the Hyperclique Hypothesis is assumed (see e.g.\ the full version of Mengel's survey~\cite[Hypothesis 3]{Mengel25arxiv}); as mentioned earlier, we omit the details in order to avoid formally introducing yet another lower bound assumption from fine-grained complexity theory.

Second, while for counting 
the answers to a CQ
in linear time the property of being acyclic is the right criterion, note that for unions of CQs, acyclicity is not even sufficient for tractability. Even when restricted to unions of acyclic conjunctive queries, the meta problem is $\mathrm{NP}$-hard.

Third, we elaborate on the idea that we use to prove Theorem~\ref{thm:main_meta}.
As mentioned before, the algorithmic part of Theorem~\ref{thm:main_meta} comes from the well-known technique of expressing UCQ answer counts in terms of linear combinations of CQ answer counts, and establishing a corresponding complexity monotonicity property, see Section~\ref{sec:complexitymon}.
The more interesting result is the hardness part. Here we discover a connection between the meta question stated in $\meta$, and a topological invariant, namely, the question whether the reduced Euler characteristic of a simplicial complex is non-zero. It is known that simplicial complexes with non-vanishing reduced Euler characteristic are evasive, and as such this property is also related to Karp's Evasiveness Conjecture (see e.g.\ the excellent survey of Miller~\cite{Miller13}). We use the known fact that deciding whether the reduced Euler characteristic is vanishing is $\mathrm{NP}$-hard~\cite{RouneS13}. 
Roughly, the reduction works as follows. Given some simplicial complex $\Delta$, we carefully define a UCQ $\Psi_\Delta$ in such a way that only one particular term in the CQ expansion of $\Psi_\Delta$ determines the linear-time tractability of counting answers to $\Psi_\Delta$. However, the coefficient of this term is zero precisely if the reduced Euler characteristic of $\Delta$ is vanishing. 

Simplicial complexes also appeared in a related context in a work by Roth and Schmitt~\cite{RothS18}. They show a connection between the complexity of counting induced subgraphs that fulfil some graph property and the question whether a simplicial complex associated with this graph property is non-zero. To solve their problem, it suffices to consider simplicial \emph{graph} complexes, which are special simplicial complexes whose elements are subsets of the edges of a complete graph, and to encode these as induced subgraph counting problems.
In contrast, 
to get our result we must encode arbitrary abstract simplicial  complexes as 
UCQs  and to show how to transfer the question about their Euler characteristic  to a question about linear-time solvability of UCQs.

It turns out that, as additional consequences of our reduction in the proof of Theorem~\ref{thm:main_meta}, we also obtain lower bounds for (approximately) computing the so-called Weisfeiler-Leman-dimension of a UCQ.

\paragraph*{Consequences for the Weisfeiler-Leman-dimension of quantifier-free UCQs} 
During the last decade we have witnessed a resurge in the study of the \emph{Weisfeiler-Leman-dimension} of graph classes and graph parameters~\cite{Arvind16,Furer17,DellGR18,KieferPS19,Morrisetal19,BarceloGMO22}. The Weisfeiler-Leman algorithm (WL-algorithm) and its higher-dimensional generalisations are  important heuristics for graph isomorphism; for example, the $1$-dimensional WL-algorithm is equivalent to the method of colour-refinement. We refer the reader to e.g.\ the EATCS Bulletin article of Arvind~\cite{Arvind16} for a concise and self-contained introduction; however, in this work we will use the WL-algorithm only in a black-box manner.

For each positive integer $k$, we say that two graphs $G_1$ and $G_2$ are $k$-\emph{WL equivalent}, denoted by $G_1 \cong_k G_2$, if they cannot be distinguished by the $k$-dimensional WL-algorithm. A graph parameter $\pi$ is called $k$\emph{-WL invariant} if 
$G_1 \cong_k G_2$ implies
$\pi(G_1)=\pi(G_2)$. 
Moreover, the \emph{WL-dimension} of $\pi$ is the minimum $k$ for which $\pi$ is $k$-WL invariant, if such a $k$ exists, and $\infty$ otherwise (see e.g.\cite{ArvindFKV22}). The WL-dimension of a graph parameter $\pi$ provides important information about the descriptive complexity of $\pi$~\cite{CaiFI92}. Moreover, recent work of Morris et al.\ \cite{Morrisetal19} shows that the WL-dimension of a graph parameter lower bounds the minimum dimension of a higher-order Graph Neural Network that computes the parameter.

The definitions of the WL-algorithm and the WL-dimension extend from graphs to labelled graphs, that is, directed multi-graphs with edge- and vertex-labels (see e.g.~\cite{LanzingerB23}). Formally, we say that a database is a \emph{labelled graph} if its signature has arity at most $2$, and if it contains no self-loops, that is, tuples of the form $(v,v)$. Similarly, \emph{(U)CQs on labelled graphs} have signatures of arity at most $2$ and contain no atom of the form $R(v,v)$.   

\begin{definition}[WL-dimension]\label{def:WLdim}
Let $\Psi$ be a UCQ on labelled graphs. The \emph{WL-dimension} of~$\Psi$, denoted by $\mathsf{dim}_{\mathrm{WL}}(\Psi)$, is the minimum $k$ such that, for any pair of labelled graphs $\calD_1$ and $\calD_2$ with $\calD_1 \cong_k \calD_2$, it holds that the number of answers to $\Psi$ in $\calD_1$ is the same as in $\calD_2$.
If no such $k$ exists, then the WL-dimension is $\infty$.
\end{definition}
\noindent Note that a CQ is a special case of a UCQ, so Definition~\ref{def:WLdim} also applies when $\Psi$ is a CQ~$\varphi$. 

It was shown very recently that the WL-dimension of a \emph{quantifier-free conjunctive query} $\varphi$ on labelled graphs is equal to the treewidth of the Gaifman graph of $\varphi$~\cite{Neuen23,LanzingerB23}. Using known algorithms for computing the treewidth~\cite{Bodlaender96,FeigeHL08} it follows that, for every fixed positive integer $d$, the problem of deciding whether the WL-dimension of $\varphi$ is at most $d$ can be solved in polynomial time (in the size of $\varphi$). Moreover, the WL-dimension of $\varphi$ can be efficiently approximated in polynomial time.

In stark contrast, we show that the computation of the WL-dimension of a \emph{UCQ} is much harder; in what follows, we say that $S$ is an $f$-\emph{approximation} of $k$ if $k \leq S \leq f(k)\cdot k$.

\begin{restatable}{theorem}{WLone}  
\label{thm:WL1}
There is an algorithm that computes a $O(\sqrt{\log k})$-approximation of the WL-dimension~$k$ of a quantifier-free UCQ on labelled graphs $\Psi=\varphi_1\vee\dots\vee\varphi_\ell$ in time $|\Psi|^{O(1)}\cdot O(2^\ell)$.

Moreover, let $f:\mathbb{Z}_{>0}\to\mathbb{Z}_{>0}$ be any computable function. 
The problem of computing an $f$-approximation of $\mathsf{dim}_{\mathrm{WL}}(\Psi)$
given an  input UCQ $\Psi=\varphi_1\vee\dots\vee\varphi_\ell$  is NP-hard, and, assuming ETH, an $f$-approximation of $\mathsf{dim}_{\mathrm{WL}}(\Psi)$ cannot be computed in time~$2^{o(\ell)}$.
\end{restatable}

Finally, the computation of the WL-dimension of UCQs stays intractable even if we fix $k$.

\begin{restatable}{theorem}{WLtwo}  
\label{thm:WL2}
Let $k$ be any \emph{fixed} positive integer. 
The problem of deciding whether the WL-dimension of a quantifier-free UCQ on labelled graphs $\Psi=\varphi_1\vee\dots\vee\varphi_\ell$ is at most $k$ can be solved in time $|\Psi|^{O(1)}\cdot O(2^\ell)$.

Moreover, the problem is $\mathrm{NP}$-hard and, assuming ETH, cannot be solved in time $2^{o(\ell)}$.
\end{restatable}

\subsection{Further Related Work}

For exact counting it makes a substantial difference whether one wants to count answers to a conjunctive query or a union of conjunctive queries~\cite{ChenM16,DellRW19}. However, for approximate counting, unions can generally be handled using a standard trick of Karp and Luby~\cite{KarpL83:focs}, and therefore fixed-parameter tractability results for approximately counting the answers to a conjunctive query also extend to unions of conjunctive queries~\cite{ArenasNew,FockeGRZ22}.

Counting and enumerating the answers to a union of conjunctive queries has also been studied in the context of dynamic databases~\cite{BerkholzKS17,BerkholzKS18}. This line of research investigates the question whether linear-time dynamic algorithms are possible. Concretely, the question is whether, after a preprocessing step that builds a data structure in time linear in the size of the initial database, the number of answers to a fixed union of conjunctive queries can be returned in constant time with a constant-time update to the data structure, whenever there is a change to the database.
Berkholz et al.\ show that for a conjunctive query such a linear-time algorithm is possible if and only if the CQ is \emph{$q$-hierarchical}~\cite[Theorem 1.3]{BerkholzKS17}. There are acyclic CQs that are not $q$-hierarchical, for instance the query $\varphi(\{a,b,c,d\})= E(a,b)\wedge E(b,c) \wedge E(c,d)$ is clearly acyclic --- however, the sets of atoms that contain $b$ and $c$, respectively, are neither comparable nor disjoint, and therefore $\varphi$ is not $q$-hierarchical. So, there are queries for which counting in the static setting is easy, whereas it is hard in the dynamic setting.
Berkholz et al.~extend their result from CQs to UCQs~\cite[Theorem 4.5]{BerkholzKS18}, where the criterion is whether the UCQ is \emph{exhaustively $q$-hierarchical}. This property essentially means that, for every subset of the CQs in the union, if instead of taking the disjunction of these CQs we take the conjunction, then the resulting CQ should be $q$-hierarchical.
Moreover, checking whether a CQ $\phi$ is $q$-hierarchical can be done in time polynomial in the size of $\phi$. However, the straightforward approach of checking whether a UCQ is exhaustively $q$-hierarchical takes exponential time, and it is stated as an open problem in \cite{BerkholzKS18} whether this can be improved. In the dynamic setting this question remains open --- however, in the static setting we show that, while for counting answers to CQs the criterion for linear-time tractability can be verified in polynomial time, this is not true for unions of conjunctive queries, subject to some complexity assumptions, as we have seen in Theorem~\ref{thm:main_meta}.

\section{Preliminaries}\label{sec:preliminaries} 
Due to the fine-grained nature of the questions we ask in this work (e.g.\ linear time counting vs non-linear time counting), it is important to specify the machine model. We use the word RAM model with $O(\log n)$ bits as one of the lower bound conjectures we rely on in this work is stated for this model (see~\cite{AbboudW14}). The exact model makes a difference. For example, it is possible to count answers to quantifier-free acyclic conjunctive queries in linear time in the word RAM model~\cite{CarmeliZBCKS22}, while Turing machines only achieve near linear time (or expected linear time)~\cite{BeraGLSS22}.

\subsection{Parameterised and Fine-grained Complexity Theory}
A \emph{parameterised counting problem} is a pair consisting of a function $P:\{0,1\}^\ast \to \mathbb{N}$ and a computable\footnote{Some authors require the parameterisation to be polynomial-time computable; see the discussion in the standard textbook of Flum and Grohe~\cite{FlumG06}.  In this work, the parameter will always be the size of the input query, which can clearly be computed in polynomial time.} parameterisation $\kappa:\{0,1\}^\ast \to \mathbb{N}$. 
For example, in the problem $\#\textsc{Clique}$ 
the function maps an input (a graph $G$ and a positive integer $k$, encoded as a string in $\{0,1\}^\ast$ to the number of $k$-cliques in $G$.
The parameter is $K$ so   $\kappa(G,k)=k$.

A parameterised counting problem $(P,\kappa)$ is called \emph{fixed-parameter tractable} (FPT) if there is a computable function $f$ and an algorithm $\mathbb{A}$ that, given input $x$, computes $P(x)$ in time $f(\kappa(x))\cdot |x|^{O(1)}$. We call $\mathbb{A}$ an \emph{FPT-algorithm} for $(P,\kappa)$.

A \emph{parameterised Turing-reduction} from $(P,\kappa)$ to $(P',\kappa')$ is an algorithm $\mathbb{A}$ equipped with oracle access to $P'$ that satisfies the following two constraints: (I) $\mathbb{A}$ is an FPT-algorithm for $(P,\kappa)$, and (II) there is a computable function $g$ such that, 
when the algorithm $\mathbb{A}$ is run with
input $x$, every oracle query $y$ to $(P',\kappa')$ has the property that 
the parameter $\kappa'(y)$   is bounded by $g(\kappa(x))$. 
We write $(P,\kappa)\fptred (P'\kappa')$ if a parameterised Turing-reduction exists.

Evidence for the non-existence of FPT algorithms is usually given by hardness for the parameterised classes $\#\ccW{1}$ and $\ccW{1}$, which can be considered to be the parameterised versions of $\#\mathrm{P}$ and $\mathrm{NP}$. The definition of those classes uses bounded-weft circuits, and we refer the interested reader e.g.\ to the standard textbook of Flum and Grohe~\cite{FlumG06}  for a comprehensive introduction. For this work, it suffices to rely on the clique problem to establish hardness for those classes:
A parameterised counting problem $(P,\kappa)$ is \emph{$\#\ccW{1}$-hard} if $\#\textsc{Clique}\fptred (P,\kappa)$, and it is \emph{$\ccW{1}$-hard} if $\textsc{Clique}\fptred (P,\kappa)$, where $\textsc{Clique}$ is the decision version of $\#\textsc{Clique}$, that is, given $G$ and $k$, the task is to decide whether there is at least one $k$-clique in $G$. 
As observed in previous works~\cite{ChenM15}, if all variables are existentially quantified, the problem of counting answers to a conjunctive query actually encodes a decision problem. So it comes to no surprise that both complexity classes $\ccW{1}$ and $\#\ccW{1}$ are relevant for its classification.
It is well known (see e.g.~\cite{Chenetal05,Chenetal06,CyganFKLMPPS15} that $\ccW{1}$-hard and $\#\ccW{1}$-hard problems are not fixed-parameter tractable, unless the Exponential Time Hypothesis fails. This hypothesis is stated as follows.

\begin{conj}[ETH~\cite{ImpagliazzoP01}]\label{conj:ETH}
$3$-SAT cannot be solved in time $\exp(o(n))$, where $n$ denotes the number of variables of the input formula.
\end{conj}

We also rely on a non-uniform version of the Exponential Time Hypothesis, which is defined below.

\begin{conj}[Non-uniform ETH~\cite{ChenEF12}]\label{conj:NUETH}
$3$-$\textsc{SAT}\notin \bigcap_{\varepsilon>0} \mathrm{DTime}(\exp(\varepsilon n))$, where $n$ denotes the number of variables of the input formula.
\end{conj}
Note that, clearly, non-uniform ETH implies ETH.

Finally, for ruling out linear-time algorithms, we will rely on the Triangle Conjecture:
\begin{conj}[Triangle Conjecture~\cite{AbboudW14}]\label{conj:Triangle}
There exists $\gamma>0$ such that any (randomised) algorithm in the word RAM model with $O(\log n)$ bits which decides whether a graph with $n$ vertices and $m$ edges contains a triangle takes time at least $\Omega(m^{1+\gamma})$ in expectation.
\end{conj}

\subsection{Structures, Homomorphisms, and Conjunctive Queries}

A \emph{signature} is a finite tuple $\tau=(R_1,\dots,R_s)$ where each $R_i$ is a relation symbol and comes with an arity $a_i$. The arity of a signature is the maximum arity of its relation symbols. A \emph{structure} $\mathcal{A}$ over $\tau$ consists of a finite universe $U(\mathcal{A})$ and a relation $R^\mathcal{A}_i$ of arity $a_i$ for each relation symbol $R_i$ of $\tau$. As usual in relational algebra, we view \emph{databases} as relational structures. We encode a structure by listing its signature, its universe and its relations. Therefore, given a structure $\mathcal{A}$ over $\tau$, we set $|\mathcal{A}|=|\tau| + |U(\mathcal{A})| + \sum_{R \in \tau} |R^\mathcal{A}|\cdot a_R$, where $a_R$ is the arity of $R$. 

For example, a \emph{graph} $G$ is a structure over the signature $(E)$ where $E$ has arity $2$.
The \emph{Gaifman graph} of a structure $\mathcal{A}$ has as vertices the universe $U(\mathcal{A})$ of $\mathcal{A}$, and for each pair of vertices $u,v$, there is  an edge $\{u,v\}$ in~$E$ if and only if at least one of the relations of $\mathcal{A}$ contains a tuple containing both $u$ and $v$.
Note that the edge set~$E$ of a Gaifman graph is symmetric and irreflexive.

Let $\mathcal{A}$ and $\mathcal{B}$ be structures over the same signature $\tau$.
Then $\mathcal{A}$ is a \emph{substructure} of $\mathcal{B}$ if $U(\mathcal{A})\subseteq U(\mathcal{B})$ and, for each relation symbol $R$ in $\tau$, it holds that $R^{\mathcal{A}}\subseteq R^{\mathcal{B}}\cap U(\mathcal{A})^a$, where $a$ is the arity of $R$.
A substructure is \emph{induced} if, for each relation symbol $R$ in $\tau$, we have $R^{\mathcal{A}}= R^{\mathcal{B}}\cap U(\mathcal{A})^a$.
A substructure $\mathcal{A}$ of $\mathcal{B}$ with
$\mathcal{A}\neq \mathcal{B}$   is a \emph{proper substructure} of $\mathcal B$.
We also define the \emph{union} $\mathcal{A}\cup \mathcal{B}$ of two structures $\mathcal{A}$ and $\mathcal{B}$ as the structure over $\tau$ with universe $U(\mathcal{A})\cup U(\mathcal{B})$ and $R^{\mathcal{A}\cup \mathcal{B}}=R^{\mathcal{A}}\cup R^{\mathcal{B}}$. Note that the union is well-defined even if the universes are not disjoint.

\paragraph{Homomorphisms as Answers to CQs}
Let $\mathcal{A}$ and $\mathcal{B}$ be structures over signatures $\tau_\mathcal{A}\subseteq \tau_\mathcal{B}$.
A \emph{homomorphism} from $\mathcal{A}$ to $\mathcal{B}$ is a mapping $h:U(\mathcal{A}) \to U(\mathcal{B})$ such that for each relation symbol $R\in \tau_\mathcal{A}$ with arity $a$ and each tuple $\Vec{t}=(t_1,\dots,t_a) \in R^{\mathcal{A}}$ we have that $h(\Vec{t})=(h(t_1),\dots,h(t_a)) \in R^{\mathcal{B}}$. We use $\Homs{\calA}{\calB}$ to denote the set of homomorphisms from $\calA$ to $\calB$, and we use the lower case version $\homs{\calA}{\calB}$ to denote the \emph{number} of homomorphisms from $\calA$ to $\calB$.

Let $\varphi$ be a conjunctive query with free variables $X=\{x_1,\dots,x_k\}$ and quantified variables $Y=\{y_1,\dots,y_d\}$. We can associate $\varphi$ with a structure $\mathcal{A}_\varphi$ defined as follows: The universe of $\mathcal{A}_\varphi$ are the variables $X\cup Y$ and for each atom $R(\Vec{t})$ of $\varphi$ we add the tuple $\Vec{t}$ to $R^\mathcal{A}$. 
It is well-known that, for each database $\mathcal{D}$, the set of answers of $\varphi$ in $\mathcal{D}$ is precisely the set of assignments $a: X \to U(\mathcal{D})$ such that there is a homomorphism $h \in\Homs{\mathcal{A}_\varphi}{\mathcal{D}}$ with $h|_{X}=a$. Since working with (partial) homomorphisms will be very convenient in this work, we will use the notation from~\cite{DellRW19} and (re)define a conjunctive query as a pair consisting of a relational structure $\mathcal{A}$ together with a set $X\subseteq U(\mathcal{A})$. The size of $(\mathcal{A},X)$ is denoted by $|(\mathcal{A},X)|$ and defined to be $|\mathcal{A}|+|X|$. Furthermore, we define 

\[ \Ans{(\mathcal{A},X)}{\mathcal{D}} := \{a: X \to U(\mathcal{D}) \mid \exists h \in\Homs{\mathcal{A}}{\mathcal{D}}: h|_{X}=a  \} \]
as the set of answers of $(\mathcal{A},X)$ in $\mathcal{D}$.
We then use $\ans{(\mathcal{A},X)}{\mathcal{D}}$ to denote the number of answers, i.e., $\ans{(\mathcal{A},X)}{\mathcal{D}}\coloneqq |\Ans{(\mathcal{A},X)}{\mathcal{D}}|$.

We can now formally define the (parameterised) problem of counting answers to conjunctive queries. As is usual, we restrict the problem by a class $C$ of allowed queries.

\vbox{
\begin{description}\setlength{\itemsep}{0pt}
\setlength{\parskip}{0pt}
\setlength{\parsep}{0pt}   			
\item[\bf Name:] $\#\homsprob(C)$ 
\item[\bf Input:]   A conjunctive query $(\mathcal{A},X)\in C$ together with a database $\mathcal{D}$.
\item[\bf Parameter:]   $|(\mathcal{A},X)|$. 
\item[\bf Output:]  The number of answers $\ans{(\mathcal{A},X)}{\mathcal{D}}$. 
\end{description}
}

\paragraph{Acyclicity and (Hyper-)Treewidth}
A structure of arity $2$ is \emph{acyclic} if its Gaifman graph is acyclic. For higher arities, a structure is acyclic if it has a join-tree or, equivalently, if it has (generalised) hyper-treewidth at most $1$ (see~\cite{Gottlob02:jcss-hypertree}). We
will not need the concept of hyper-treewidth, but we refer the reader to~\cite{GottlobGS14} for a comprehensive treatment.

The treewidth of graphs and structures is defined as follows
\begin{definition}[Tree decompositions, treewidth]
    Let $G$ be a graph. A \emph{tree decomposition} of $G$ is a pair $(T,B)$, where $T$ is a (rooted) tree, and $B$ assigns each vertex $t\in V(T)$ a \emph{bag} $B_t$ such that the following constraints are satisfied:
    \begin{itemize}
        \item[(C1)] $V(G)= \bigcup_{t\in V(T)} B_t$,
        \item[(C2)] For each edge $e\in E(G)$ there exists $t\in V(T)$ such that $e\subseteq B_t$, and
        \item[(C3)] For each $v\in V(G)$, the subgraph of $T$ containing all vertices $t$ with $v\in B_t$ is connected.
    \end{itemize}
    The \emph{width} of a tree decomposition is $\max_{t\in V(T)} |B_t| -1$, and the \emph{treewidth} of $G$ is the minimum width of any tree decomposition of $G$. Finally, the treewidth of a structure is the treewidth of its Gaifman graph.
\end{definition}

\paragraph{\#Equivalence, \#Minimality, and \#Cores}
In the realm of decision problems, it is well known that evaluating a conjunctive query is equivalent to evaluating the (homomorphic) core of the query, i.e., evaluating the minimal homomorphic-equivalent query (see, for instance, the seminal work of Chandra and Merlin~\cite{ChandraM77}). A similar, albeit slightly different notion of equivalence and minimality is required for counting answers to conjunctive queries. In what follows, we will provide   the necessary definitions and properties of equivalence, minimality and cores for counting answers to conjunctive queries, and we refer the reader to~\cite{ChenM16} and to the full version of~\cite{DellRW19} for a more comprehensive discussion. To avoid confusion between the notions in the realms of decision and counting, we will from now on use the \# symbol for the counting versions (see Definition~\ref{def:cequivalence}).

An \emph{isomorphism} from a structure $\mathcal{A}$ to a structure $\mathcal{A}'$ is a bijection $b: U(\mathcal{A}) \to U(\mathcal{A}')$ such that for each relation symbol $R \in \tau_A$ and each tuple $\vec{t}$ of elements of $U(\mathcal{A})$, we have $\vec{t}\in R^\mathcal{A}$ if and only if $b(\vec{t})\in R^{\mathcal{A}'}$. 
\begin{definition}
Two conjunctive queries $(\mathcal{A},X)$ and $(\mathcal{A}',X')$ are \emph{isomorphic}, denoted by $(\mathcal{A},X)\cong (\mathcal{A}',X')$, if there is an isomorphism $b$ from $\mathcal{A}$ to $\mathcal{A}'$ with $b(X)=X'$.
\end{definition}

\begin{definition}[\#Equivalence and \#minimality (see~\cite{ChenM16,DellRW19})]\label{def:cequivalence}
Two conjunctive queries $(\mathcal{A},X)$ and $(\mathcal{A}',X')$ are \emph{$\cequivalent$}, denoted by $(\mathcal{A},X)\sim (\mathcal{A}',X')$, if for every database $\mathcal{D}$ we have $\ans{(\mathcal{A},X)}{\mathcal{D}}=\ans{(\mathcal{A}',X')}{\mathcal{D}}$. A conjunctive query $(\mathcal{A},X)$ is \emph{$\cminimal$} if there is no proper substructure $\mathcal{A}'$ of $\mathcal{A}$ such that $(\mathcal{A},X)\sim (\mathcal{A}',X)$.
\end{definition}

\begin{observation}\label{obs:cminimalequivalent}
The following are equivalent:
\begin{enumerate}
\item A conjunctive query $(\mathcal{A},X)$ is \emph{$\cminimal$}.
\item $(\mathcal{A},X)$ has no $\cequivalent$ substructure that is induced by a set $U$ with $X\subseteq U\subset U(\mathcal{A})$.
\item Every homomorphism from $\mathcal{A}$ to itself 
that is the identity on~$X$ is surjective.
\end{enumerate}
\end{observation}

It turns out that \#equivalence is the same as isomorphism if all variables are free, and it is the same as homomorphic equivalence 
if all variables are existentially quantified (see e.g.\ the discussion in Section 5 in the full version of~\cite{DellRW19}). 
Moreover, each quantifier-free conjunctive query is $\cminimal$.

\begin{lemma}[Corollary 54 in full version of~\cite{DellRW19}]\label{lem:counting_minimal_isomorphic}\label{lem:Z}
    If $(\mathcal{A},X)$ and $(\mathcal{A}',X')$ are $\cminimal$ and $(\mathcal{A},X)\sim(\mathcal{A}',X')$ then $(\mathcal{A},X)\cong (\mathcal{A}',X')$.
\end{lemma}

As isomorphism trivially implies \#equivalence, Lemma~\ref{lem:counting_minimal_isomorphic} shows that, for $\cminimal$ queries, \#equivalence and isomorphism coincide.

\begin{definition}[$\ccore$]\label{def:core}
A \emph{$\ccore$} of a conjunctive query $(\mathcal{A},X)$ is a $\cminimal$ conjunctive query $(\mathcal{A}',X')$ with $(\mathcal{A},X) \sim (\mathcal{A}',X')$.
\end{definition}
By Lemma~\ref{lem:counting_minimal_isomorphic}, the $\ccore$ is unique up to isomorphisms; in fact, this allows us to speak of ``the'' $\ccore$ of a conjunctive query.

\paragraph{Classification of $\#\homsprob(C)$ via Treewidth and Contracts}
It is well known that the complexity of counting answers to a conjunctive query is governed by its treewidth, and by the treewidth of its contract~\cite{ChenM15,DellRW19}, which we define as follows.
\begin{definition}[Contract]\label{def:contract}
Let $(\mathcal{A},X)$ be a conjunctive query, let $Y=U(\mathcal{A})\setminus X$, and let $G$ be the Gaifmann graph of $\mathcal{A}$. The \emph{contract} of $(\mathcal{A},X)$, denoted by $\contract{\mathcal{A},X}$ is obtained from $G[X]$ by adding an edge between each pair of vertices $u$ and $v$ for which there is a connected component $S$ in $G[Y]$ that is adjacent to both $u$ and $v$, that is, there are vertices $x,y\in S$ such that $\{x,u\}\in E(G)$ and $\{y,v\}\in E(G)$. Given a class of conjunctive queries $C$, we write $\contract{C}$ for the class of all contracts of queries in $C$.
\end{definition}

We note that there are multiple equivalent ways to define the contract of a query. For our purposes, the definition in~\cite{DellRW19} is most suitable. Also, the treewidth of the contract of a conjunctive query is an upper bound of what is called the query's ``star size'' in~\cite{DurandM13} and  its ``dominating star size'' in~\cite{DellRW19}.

Chen and Mengel established the following classification for counting answers to conjunctive queries of bounded arity.
\begin{theorem}[\cite{ChenM15}]\label{thm:cq_classification}
    Let $C$ be a recursively enumerable class of conjunctive queries of bounded arity, and let $C'$ be the class of $\ccores$ of queries in $C$. If the treewidth of $C'$ and of $\contract{C'}$ is bounded, then $\#\homsprob(C)$ is solvable in polynomial time. Otherwise, $\#\homsprob(C)$ is $\ccW{1}$-hard.
\end{theorem}
We point out that the $\ccW{1}$-hard cases can further be partitioned into $\ccW{1}$-complete, $\#\ccW{1}$-complete and even harder cases~\cite{ChenM15,DellRW19}.\footnote{Those cases are: $\#\ccW{2}$-hard and $\#\mathrm{A}[2]$-complete.} However, for the purpose of this work, we are only interested in tractable and intractable cases (recall that $\ccW{1}$-hard problems are not fixed-parameter tractable under standard assumptions from fine-grained and parameterised complexity theory, such as ETH).

\paragraph{Self-join-free Conjunctive Queries and Isolated Variables}
A conjunctive query $(\mathcal{A},X)$ is \emph{self-join-free} if each relation of $\mathcal{A}$ contains at most one tuple. We say that a variable of a conjunctive query is \emph{isolated} if it is not part of any relation.

Note that that adding/removing isolated variables to/from a conjunctive query does not change its treewidth or the treewidth of its $\ccore$. Further, it does not change the complexity of counting answers: Just multiply/divide by $n^v$, where $n$ is the number of elements of the database and $v$ it the number of added/removed isolated free variables. For this reason, we will allow ourselves in this work to freely add and remove isolated variables from the queries that we encounter.
For the existence of homomorphisms we also observe the following. 
\begin{observation}\label{obs:CQplusisolated}
Let $(\mathcal{A},X)$ be a conjunctive query, let $X'$ be a superset of $X$ and let $\mathcal{A}'$ be the structure obtained from $\mathcal{A}$ by adding an isolated variable for each $x\in X'\setminus X$. Then for all $a: X' \to U(\mathcal{D})$ we have that $a|_{X}\in\Ans{(\mathcal{A},X)}{\mathcal{D}}$ iff $a\in\Ans{(\mathcal{A}',X')}{\mathcal{D}}$.
\end{observation}

\subsection{Unions of CQs and the Homomorphism Basis}\label{sec:UCQstoHoms}

A \emph{union of conjunctive queries} (UCQ) $\Psi$ is a tuple of structures $(\calA_1, \ldots, \calA_{\ell(\Psi)})$ over the same signature 
together with a set of designated elements $X$ 
(the free variables)
that are in the universe of each of the structures. 
For each $J\subseteq [\ell(\Psi)]$, we define $\Psi\vert_{J}=((A_j)_{j\in J}, X)$.
If we restrict to a single term of the union then we usually just write $\Psi_i$ instead of $\Psi\vert_{\{i\}}$. Note that $\Psi_i=(\calA_i, X)$ is simply a conjunctive query (rather than a union of CQs). 

We will assume (without loss of generality) that, for any distinct $i$ and $i'$ in $[\ell(\Psi)]$, 
$U(\mathcal{A}_i) \cap U(\mathcal{A}_{i'}) = X$, i.e., that each CQ in the union has its own set of existentially quantified variables.

If each such conjunctive query is acyclic we say that $\Psi$ is a union of acyclic conjunctive queries. Moreover, the arity of $\Psi$ is the maximum arity of any of the $\calA_i$. The size of $\Psi$ is $|\Psi|=\sum_{i=1}^{\ell(\Psi)} |\Psi_i|$. The elements of $X$ are the \emph{free variables} of $\Psi$ and $\ell(\Psi)$ is the number of CQs in the union.
    
The set of \emph{answers} of $\Psi$ in a database $\calD$, denoted by $\Ans{\Psi}{\calD}$ is defined as follows:
\[\Ans{\Psi}{\calD} = \left\{a: X \to U(\calD) \mid \exists i \in[\ell]: a\in
\Ans{\Psi_i}{\calD}\right\} .\]
Again, we use the lower case version $\ans{\Psi}{\calD}$ to denote the \emph{number} of answers of $\Psi$ in $\calD$.

In the definition of UCQs we assume that every CQ in the union has the same set of free variables, namely~$X$. This assumption is without loss of generality. To see this, suppose that we have a union of CQs $(\mathcal{A}_1,X_1),\dots,(\mathcal{A}_\ell,X_\ell)$ with individual sets of free variables. Let $X=\bigcup_{i=1}^\ell X_i$ and, for each $i\in[\ell]$, let $\mathcal{A}'_i$ be the structure obtained from $\mathcal{A}_i$ by adding an isolated variable  
for each $x\in X\setminus X_i$. Then consider the UCQ
$ \Psi:=((\mathcal{A}'_1,\ldots,\mathcal{A}'_\ell),X)$.
If for some assignment $a: X \to U(\mathcal{D})$ it holds that there is an $i\in [\ell]$ such that $a|_{X_i} \in  \Ans{(\mathcal{A}_i,X_i)}{\mathcal{D}}$, then, according to Observation~\ref{obs:CQplusisolated}, this is equivalent to $a \in  \Ans{(\mathcal{A}'_i,X)}{\mathcal{D}}$, which means that $a$ is an answer of $\Psi$.
So, without loss of generality we can work with $\Psi$, which uses the same set of free variables for each CQ in the union.

Now we define the parameterised problem of counting answers to UCQs. As usual, the problem is restricted by a class $C$ of allowed queries with respect to which we classify the complexity.

\vbox{
\begin{description}\setlength{\itemsep}{0pt}
\setlength{\parskip}{0pt}
\setlength{\parsep}{0pt}   			
\item[\bf Name:] $\#\ucq(C)$ 
\item[\bf Input:]   A UCQ $\Psi\in C$ together with a database $\mathcal{D}$.
\item[\bf Parameter:]   $|\Psi|$. 
\item[\bf Output:]  The number of answers $\ans{\Psi}{\mathcal{D}}$. 
\end{description}
}

The next definition will be crucial for the analysis of the complexity of $\#\ucq(C)$.
\begin{definition}[combined query $\conjunctfull{\Psi}$]
    Let $\Psi=((\mathcal{A}_1,,\dots,\mathcal{A}_\ell),X)$ be a UCQ. Then we define the \emph{combined query} $\conjunctfull{\Psi}=(\bigcup_{j\in [\ell]}A_j,X)$.
\end{definition}
What follows is an easy, but crucial observation about $\conjuncts{\Psi}{J}$.
\begin{observation}\label{fact:easy_fact}
Let $((\mathcal{A}_1,,\dots,\mathcal{A}_\ell),X)$ be a UCQ, and let $\emptyset\neq J\subseteq[\ell]$. For each database $\mathcal{D}$ and assignment $a: X \to U(\mathcal{D})$ we have
\[ a \in \Ans{\conjuncts{\Psi}{J}}{\mathcal{D}} \Leftrightarrow \forall j \in J : a \in \Ans{\Psi_j}{\mathcal{D}}.\]
\end{observation}

\begin{definition}[Coefficient function $c_\Psi$]\label{def:coefUCQ}
Let $\Psi = ((\mathcal{A}_1,,\dots,\mathcal{A}_\ell),X)$ be a UCQ. For each conjunctive query $(\mathcal{A},X)$, we set $\mathcal{I}(\mathcal{A},X) = \{J \subseteq [\ell] \mid  (\mathcal{A},X) \cequiv \conjuncts{\Psi}{J}\}$, and we define the \emph{coefficient function} of $\Psi$ as follows: 
    \[ c_\Psi(\mathcal{A},X) = \sum_{J \in \mathcal{I}(\mathcal{A},X)} (-1)^{|J|+1} \,.\]
\end{definition}

Using inclusion-exclusion, we can transform the problem of counting answers to $\Psi$ 
into the problem of evaluating
a linear combination of CQ answer counts. We include a proof only for reasons of self-containment and note that the complexity-theoretic applications of this transformation, especially regarding lower bounds, have first been discovered by Chen and Mengel~\cite{ChenM16}.

\begin{lemma}[\cite{ChenM16}]\label{lem:UCQ_to_hombasis}
Let $\Psi = ((\mathcal{A}_1,,\dots,\mathcal{A}_\ell),X)$ be a UCQ.  For every database~$\mathcal{D}$,
\[ \ans{\Psi}{\mathcal{D}} = \sum_{(\mathcal{A},X)} c_\Psi(\mathcal{A},X) \cdot \ans{(\mathcal{A},X)}{\mathcal{D}},\]
where the sum is over all equivalence classes of $\cequiv$.
\end{lemma}
\begin{proof}
By inclusion-exclusion and Observation~\ref{fact:easy_fact},
\begin{align*}
\ans{\Psi}{\mathcal{D}} &= \sum_{\emptyset \neq J \subseteq [\ell]} (-1)^{|J|+1} \cdot \abs{\{a: X \to U(\mathcal{D}) \mid \forall j \in J : a \in \Ans{(A_j,X)}{\mathcal{D}} \}}\\
~&=\sum_{\emptyset \neq J \subseteq [\ell]} (-1)^{|J|+1} \cdot \ans{\conjuncts{\Psi}{J}}{\mathcal{D}}. 
\end{align*}
The claim then follows by collecting $\cequivalent$ terms.
\end{proof}

We conclude this subsection with the following two operations on classes of UCQs.
\begin{definition}[$\Gamma(C)$ and $\conjunctclass{C}$]\label{def:Gamma}
    Let $C$ be a class of UCQs.   $\Gamma(C)$ is the class of all 
        $(\mathcal{A},X)$ such that $(\mathcal{A},X)$ is $\cminimal$ and there is  $\Psi\in C$ with $c_\Psi(\mathcal{A},X) \neq 0$.
    Let $\conjunctclass{C} = \{ \conjunctfull{\Psi} \mid \Psi\in C \}$.
\end{definition}

\subsection{Complexity Monotonicity}\label{sec:complexitymon}
The principle of Complexity Monotonicity states that the computation of a linear combination of homomorphism counts is precisely as hard as computing its hardest term. It was independently discovered by Curticapean, Dell and Marx~\cite{CurticapeanDM17} and by Chen and Mengel~\cite{ChenM16}. Moreover, Chen and Mengel established the principle in the more general context of linear combinations of conjunctive queries. Formally, their result is stated below for the special case of linear combinations derived via counting answers to UCQs (see Lemma~\ref{lem:UCQ_to_hombasis}). We include a proof since we need to pay some special attention to running times, as we will also be interested in the question of linear time tractability. 

\begin{theorem}[Implicitly by~\cite{ChenM16}]
There is an algorithm $\mathbb{A}$ with the following properties:
\begin{enumerate}
\item The input of $\mathbb{A}$ is a UCQ $\Psi$ and a database $\mathcal{D}$.
\item $\mathbb{A}$ has oracle access to the function
$ \mathcal{D}' \mapsto \ans{\Psi}{\mathcal{D}'}$.
\item The output of $\mathbb{A}$ is a list with entries $((\mathcal{A},X),\ans{(\mathcal{A},X)}{\mathcal{D}})$ for each $(\mathcal{A},X)$ in the support of $c_\Psi$.
\item $\mathbb{A}$ runs in time $f(|\Psi|) \cdot O(|\mathcal{D}|)$ for some computable function $f$.
\end{enumerate}
\end{theorem}
\begin{proof}
    A crucial operation in the construction is the Tensor product of relational structures. Let $\mathcal{A}$ and $\mathcal{B}$ be structures over the signatures $\tau_A$ and $\tau_B$. 
    The structure $\mathcal{A}\otimes \mathcal{B}$ is defined as follows: The signature is $\tau_A \cap \tau_B$, and the universe is $U(\mathcal{A})\times U(\mathcal{B})$. Moreover, for every relation symbol $R \in \tau_A \cap \tau_B$ with arity $r$, a tuple $((u_1,v_1),\dots,(u_r,v_r))$ is contained in $R^{\mathcal{A}\otimes\mathcal{B}}$ if and only if $(u_1,\dots,u_r)\in R^{\mathcal{A}}$ and $(v_1,\dots,v_r)\in R^{\mathcal{B}}$.

Observe that $\mathcal{A}\otimes \mathcal{B}$ is of size bounded by and can be computed in time $O(|\mathcal{A}||\mathcal{B}|)$.\footnote{Compute the Cartesian product of $U(\mathcal{A})$ and $U(\mathcal{B})$ and then, for every relation $R\in \tau_A \cap \tau_B$ iterate over all pairs of tuples in $R^\mathcal{A}$ and $R^\mathcal{B}$ and add their point-wise product to $R^{\mathcal{A}\otimes \mathcal{B}}$.}
The algorithm $\mathbb{A}$ proceeds as follows. Let $\Psi=((\mathcal{A}_1,\dots,\mathcal{A}_\ell),X)$ be the input. For a selected set of structures $\mathcal{B}_1\dots,\mathcal{B}_k$, specified momentarily, the algorithm queries the oracle on the Tensor products $\mathcal{D} \otimes \mathcal{B}_i$. Using Lemma~\ref{lem:UCQ_to_hombasis}, this yields the following equations:

\[\ans{\Psi}{\mathcal{D}\otimes \mathcal{B}_i} = \sum_{(\mathcal{A},X)} c_\Psi(\mathcal{A},X) \cdot \ans{(\mathcal{A},X)}{\mathcal{D}\otimes \mathcal{B}_i}.\]
Next, we use the fact (see e.g.\ \cite{ChenM16}) that the Tensor product is multiplicative with respect to counting answers to conjunctive queries:
\[ \ans{(\mathcal{A},X)}{\mathcal{D}\otimes \mathcal{B}_i} = \ans{(\mathcal{A},X)}{\mathcal{D}}\cdot \ans{(\mathcal{A},X)}{\mathcal{B}_i}.\]
In combination, the previous equations yield a system of linear equations:
\[\ans{\Psi}{\mathcal{D}\otimes \mathcal{B}_i} = \sum_{(\mathcal{A},X)} c_\Psi(\mathcal{A},X) \cdot \ans{(\mathcal{A},X)}{\mathcal{D}}\cdot \ans{(\mathcal{A},X)}{\mathcal{B}_i},\]
the unknowns of which are $c_\Psi(\mathcal{A},X) \cdot \ans{(\mathcal{A},X)}{\mathcal{D}}$. It was shown in~\cite{ChenM16} and~\cite{DellRW19} that it is always possible to find $\mathcal{B}_i$ for which the system is non-singular. Moreover, the time it takes to find the $\mathcal{B}_i$ only depends on $\Psi$. Finally, solving the system yields the terms $c_\Psi(\mathcal{A},X) \cdot \ans{(\mathcal{A},X)}{\mathcal{D}}$ from which we can recover $\ans{(\mathcal{A},X)}{\mathcal{D}}$ by dividing by $c_\Psi(\mathcal{A},X)$. The overall running time is bounded by $f(|\Psi|) \cdot O(|\mathcal{D}|)$ for some computable function $f$: to see this, note that the numbers in the system can be generously bounded by $f'(|\Psi|)\cdot |\mathcal{D}|^{|\Psi|}$, hence their size is bounded by $\log(f'(|\Psi|)) + |\Psi|\log(|\mathcal{D}|)$ and thus the arithmetic operations needed for solving the system do not create a super-linear overhead in~$|\mathcal{D}|$. This concludes the proof.
\end{proof}

\begin{corollary}\label{cor:UCQ_monotone}
Let $\Psi$ be a UCQ. For each $d\geq 1$, computing the function $\calD\mapsto \ans{\Psi}{\calD}$ can be done in time $O(\abs{\calD}^d)$ if and only if for each $\cminimal$ $(\mathcal{A},X)$ with $c_\Psi(\mathcal{A},X)\neq 0$ the function $\calD \mapsto \ans{(\mathcal{A},X)}{\calD}$ can be computed in time $O(\abs{\calD}^d)$.
\end{corollary}

\begin{corollary}[Implicitly also in \cite{ChenM16}]\label{cor:UCQ_monotone_dichotomy}
Let $C$ be a recursively enumerable class of UCQs. The problems $\#\ucq(C)$ and $\#\homsprob(\Gamma(C))$ are interreducible with respect to parameterised Turing-reductions.
\end{corollary}

As a consequence, in combination with Theorem~\ref{thm:cq_classification}, Chen and Mengel~\cite{ChenM16} established the following dichotomy, which we can state in a convenient way using our notation:
\ChenMUCQ*

In other words, as a summary of the above results, we know that counting answers to a UCQ is hard if and only if a conjunctive query survives in the CQ expansion, whose $\ccore$ has either high treewidth or its contract has high treewidth. Unfortunately, given a concrete UCQ, or a concrete class of UCQs, it is not clear how to determine whether such high treewidth terms survive. 

The central question that we ask in this work is whether this implicit tractability criterion for counting answers to UCQs can be rephrased as a more natural one. Equivalently, this means that we aim to find out whether there is a natural criterion for the existence of high treewidth terms in the CQ expansion. Such natural criteria have been found for subgraph and induced subgraph counting\footnote{(Induced) Subgraph counting is a special case of counting answers to quantifier-free conjunctive queries with inequalities (and negations).}~\cite{CurticapeanDM17,FockeR22}, and for conjunctive queries with inequalities and negations~\cite{DellRW19,Roth21}.

\section{Classifications for Deletion-Closed UCQs}

Let $C$ be a class of UCQs.
Recall that $\conjunctclass{C}$ is the class of all conjunctive queries that are obtained just by substituting all $\vee$ by $\wedge$ in UCQs in $C$, whereas $\Gamma(C)$ in Theorem~\ref{thm:implicit_UCQ_dicho} is the much less natural class of $\cminimal$ queries that survive with a non-zero coefficient in the CQ expansion of a UCQ in~$C$.
The work of Chen and Mengel~\cite{ChenM16} implicitly also shows an upper bound for counting answers to UCQs from the class $C$ in terms of the simpler objects $\conjunctclass{C}$ and $\contract{\conjunctclass{C}}$, rather than in terms of the more complicated objects $\Gamma(C)$ and $\contract{\Gamma(C)}$.  
We include a proof for completeness.

\begin{lemma}\label{lem:UCQ_algo}
Let $C$ be a recursively enumerable class of UCQs. Suppose that both $\conjunctclass{C}$ and $\contract{\conjunctclass{C}}$ have bounded treewidth. Then $\#\ucq(C)$ is fixed-parameter tractable.
\end{lemma}
\begin{proof}
Let $\Psi \in C$. Recall from the proof of Lemma~\ref{lem:UCQ_to_hombasis} that, for every databse $\calD$,
\begin{align*}
    \ans{\Psi}{\calD} =\sum_{\emptyset \neq J \subseteq [\ell]} (-1)^{|J|+1} \cdot \homs{\conjuncts{\Psi}{J}}{\calD}\,. 
\end{align*}
Hence $\#\ucq(C)\fptred\#\homsprob(\hat{C})$ where $\hat{C}$ is 
$\{\conjuncts{\Psi}{J} \mid \Psi \in C ~\wedge~\emptyset \neq J \subseteq [\ell(\Psi)] \}$.
Finally, since $\conjuncts{\Psi}{J}$ is a subquery of $\conjunctfull{\Psi}$ for each $J$, the treewidths of $\conjuncts{\Psi}{J}$ and $\contract{\conjuncts{\Psi}{J}}$ are 
bounded from above by the treewidths of $\conjunctfull{\Psi}$ and $\contract{\conjunctfull{\Psi}}$, respectively. Consequently, the treewidths of $\hat{C}$ and $\contract{\hat{C}}$ are bounded, and thus $\#\homsprob(\hat{C})$ is polynomial-time solvable by the classification of Chen and Mengel~\cite[Theorem 22]{ChenM15} 
Since $\#\ucq(C)\fptred\#\homsprob(\hat{C})$,  the lemma follows.
\end{proof} 

Our goal is to relate the complexity of $\#\ucq(C)$ to the structure of $\conjunctclass{C}$ with the hope of obtaining a more natural tractability criterion than the one given by Theorem~\ref{thm:implicit_UCQ_dicho}.
While we will see later that this seems not always possible (Appendix~\ref{sec:appendix}), we identify conditions under which a natural criterion based on $\conjunctclass{C}$ is possible, both in the quantifier-free case (Section~\ref{sec:Meta_quantfreeLB}), and in the general case that allows quantified variables (Section~\ref{sec:Meta_generalLB}).

A class of UCQs $C$ is \emph{closed under deletions} if, for every $\Psi=((\calA_1,\ldots,\calA_\ell),X)\in C$ and for every $\emptyset\neq J\subseteq [\ell]$, the UCQ $\Psi\vert_{J}$
is also contained in $C$. For example, any class of UCQs defined by prescribing the allowed conjunctive queries is closed under deletions. This includes, e.g., unions of acyclic conjunctive queries.

\subsection{The Quantifier-free Case}\label{sec:Meta_quantfreeLB}
 As a warm-up, we start with the much simpler case of quantifier-free queries. Here, we only allow (unions of) conjunctive queries $(\mathcal{A},X)$ satisfying $U(\mathcal{A})=X$.

\begin{lemma}\label{lem:main_dichotomy_quantifierfree}
Let $C$ be a recursively enumerable class of quantifier-free UCQs of bounded arity. Suppose that $C$ is closed under deletions. If $\conjunctclass{C}$ has unbounded treewidth then $\#\ucq(C)$ is $\ccW{1}$-hard.
\end{lemma}
\begin{proof}
We show that $\conjunctclass{C}\subseteq \Gamma(C)$, which then proves the claim by Theorem~\ref{thm:implicit_UCQ_dicho}. Recall from Definition~\ref{def:Gamma} that $\Gamma(C)= \{(\mathcal{A},X) \mid (\mathcal{A},X) \text{ is $\cminimal$ and there is } \Psi\in C \text{ with } c_\Psi(\mathcal{A},X) \neq 0 \}$. Let  $\Psi=((\calA_1,\ldots,\calA_\ell),X) \in C$. Note that, according to Observation~\ref{obs:cminimalequivalent}, $\conjunctfull{\Psi}$ is its own $\ccore$ since  it does not have existentially quantified variables. 
For the same reason, for each nonempty subset~$J$ of~$[\ell]$,   the query $\conjuncts{\Psi}{J}$ is its own $\ccore$. Now let $J\subseteq [\ell]$ be inclusion-minimal with the property that $\conjuncts{\Psi}{J}$ is isomorphic to $\conjunctfull{\Psi}$. Since $C$ is closed under deletions, the UCQ $\Psi\vert_{J}$ is contained in $C$. By the inclusion-minimality of $J$, Definition~\ref{def:coefUCQ} ensures that   $c_{\Psi\vert_{J}}(\conjunctfull{\Psi}) =(-1)^{|J|+1} \neq 0$. As a consequence, $\conjunctfull{\Psi}\in \Gamma(C)$, concluding the proof.
\end{proof}

From Lemmas~\ref{lem:main_dichotomy_quantifierfree} and~\ref{lem:UCQ_algo} together with the fact that the contract of a quantifier-free query is the query itself, we obtain the classification for quantifier-free UCQs, which we restate for convenience.
\maindichosimple*

\subsection{The General Case}\label{sec:Meta_generalLB}

Now we consider UCQs with existentially quantified variables. Here, a corresponding hardness result (Lemma~\ref{lem:main_dichotomy_selfjoinfree}) can be achieved under some additional assumptions. 
Note that the number of existentially quantified variables in a UCQ $\Psi=((\mathcal{A}_1,\ldots,\mathcal{A}_\ell),X)$ is equal to $\sum_{i=1}^\ell |U(\mathcal{A}_i) \setminus X|$. 
We first need the following two auxiliary results:

\begin{lemma}\label{lem:ccore_subraph}\label{lem:new}
    Let $(\mathcal{A},X)$ and $(\mathcal{A}',X')$ be $\cequivalent$ conjunctive queries. Further, let $G$ and $G'$ be the Gaifman graphs of $\mathcal{A}$ and $\mathcal{A}'$, respectively. Then $G[X]$ and $G'[X']$ are isomorphic.
\end{lemma}
\begin{proof}
By~\cite{ChenM16} (see Lemma 48 in the full version of~\cite{DellRW19} for an explicit statement),   there are surjective functions $s: X \to X'$ and $s': X' \to X$ and homomorphisms $h\in \Homs{\mathcal{A}}{\mathcal{A}'}$ and $h'\in \Homs{\mathcal{A}'}{\mathcal{A}}$ such that $h|_X = s$ and $h'|_{X'}=s'$. 

Clearly, $s$ and $s'$ are bijective. Let $e=\{u,v\}$ be an edge of $G[X]$. Then there exists a tuple $\vec{t}$ of elements of $U(\mathcal{A})$ such that
\begin{itemize}
        \item[(i)] $\vec{t}\in R^{\mathcal{A}}$ for some relation (symbol) $R$ of the signature of $\mathcal{A}$, and
        \item[(ii)] $u$ and $v$ are elements of $\vec{t}$.
\end{itemize}
Since $h$ is a homomorphism, $h(\vec{t})\in R^{\mathcal{B}}$. Thus $\{h(u),h(v)\} = \{s(u),s(v)\}$ is an edge of $G'[X']$. The backward direction is analogous.
\end{proof}

\begin{lemma}\label{lem:self-join-free-counting-cores} 
    Let $(\mathcal{A},X)$ be a self-join-free conjunctive query. Let $\mathcal{A}'$ be the structure obtained from $\mathcal{A}$ by deleting all isolated variables in $U(\mathcal{A})\setminus X$. Then $(\mathcal{A}',X)$ is the $\ccore$ of $(\mathcal{A},X)$.
\end{lemma}
\begin{proof}
 
Clearly, $(\mathcal{A},X)$ and $(\mathcal{A}',X)$ are $\cequivalent$. Thus it remains to show that $(\mathcal{A}',X)$ is $\cminimal$. Assume for contradiction by Observation~\ref{obs:cminimalequivalent} that 
$(\mathcal{A}',X)$ has an \#equivalent substructure $\hat{\mathcal{A}}$ that is induced by a set $U$ with 
$X \subseteq U \subset U(\mathcal{A}')$.
 
Since $\mathcal{A}'$ is self-join-free and it does not have isolated variables,
there is a relation (symbol) $R$ such that $R^{\mathcal{A}'}$ contains precisely one tuple, and $R^{\hat{\mathcal{A}}}$ is empty. Thus, there is no homomorphism from $\mathcal{A}'$ to $\hat{\mathcal{A}}$ and, consequently, $(\hat{\mathcal{A}},X)$ and $(\mathcal{A}',X)$ cannot be $\cequivalent$, yielding a contradiction and concluding the proof.
\end{proof}

\begin{lemma}\label{lem:main_dichotomy_selfjoinfree}
Let $C$ be a recursively enumerable class of unions of self-join-free conjunctive queries with bounded arity. Suppose that $C$ is closed under deletions and that there is a finite upper bound on the number of existentially quantified variables in queries in $C$. If either of $\conjunctclass{C}$ or $\contract{\conjunctclass{C}}$ have unbounded treewidth then $\#\ucq(C)$ is $\ccW{1}$-hard.
\end{lemma}
\begin{proof}
Let $d$ be the maximum number of existentially quantified variables in a query in $C$.
Assume first that $\conjunctclass{C}$ has unbounded treewidth. We show that $\Gamma(C)$ has unbounded treewidth, which proves the claim by Theorem~\ref{thm:implicit_UCQ_dicho}. 
To this end, let $B$ be any positive integer. The goal is to find a conjunctive query in $\Gamma(C)$ with treewidth at least $B$. Since $\conjunctclass{C}$ has unbounded treewidth, there is a UCQ $\Psi=((\mathcal{A}_1,\ldots,\mathcal{A}_\ell),X)$ in $C$ such that $\conjunctfull{\Psi}$ has treewidth larger than $d+ B$. Note that, although all $\Psi_i$ are self-join-free, $\conjunctfull{\Psi}$ is not necessarily self-join-free. Let $J$ be inclusion-minimal among the subsets of~$[\ell]$ with the property that the $\ccore$ of $\conjuncts{\Psi}{J}$ is isomorphic to the $\ccore$ of $\conjunctfull{\Psi}$.  Since $C$ is closed under deletions, the UCQ $\Psi\vert_{J}$ is contained in $C$. Let $(\mathcal{A}',X')$ be 
the $\ccore$ of $\conjunctfull{\Psi}$.

By inclusion-minimality of $J$,   $c_{\Psi\vert_{J}}((\mathcal{A}',X')) =(-1)^{|J|+1} \neq 0$. As a consequence, $(\mathcal{A}',X')\in \Gamma(C)$. It remains to show that the treewidth of $(\mathcal{A}',X')$ is at least $B$. For this, let $G$ be the Gaifmann graph of $\conjunctfull{\Psi}$ and let $G'$ be the Gaifmann graph of the $\ccore$ of $\conjunctfull{\Psi}$ (the Gaifman graph of $\mathcal{A}'$). First, deletion of a vertex can decrease the treewidth by at most $1$. Thus, $G[X]$ has treewidth at least $d+B-d = B$. By Lemma~\ref{lem:ccore_subraph}, $G[X]$ and $G'[X']$ are isomorphic.

Therefore the treewidth of $G'[X']$, i.e., the treewidth of $(\mathcal{A}',X')$, is at least $B$.
So we have shown that if the treewidth of $\conjunctclass{C}$ is unbounded then so is the treewidth of $\Gamma(C)$.

In the second case, we assume that the contracts of queries in~$\conjunctclass{C}$ (see Definition~\ref{def:contract}) have unbounded treewidth. We introduce the following terminology: Let $(\mathcal{A},X)$ be a conjunctive query and let $y\in U(\mathcal{A})$. The \emph{degree of freedom} of $y$ is the number of vertices in $X$ that are adjacent to $y$ in the Gaifman graph of $\mathcal{A}$. Let $\hat{C}$ be the class of all conjunctive queries $(\mathcal{A},X)$ such that there exists $\Psi=((\mathcal{A}_1,\ldots,\mathcal{A}_{\ell(\Psi)}),X)$ in $C$ with $(\mathcal{A},X)=(\mathcal{A}_i,X)$ for some $i\in[\ell(\Psi)]$. Since $C$ is closed under deletions, $\hat{C}\subseteq C$. By the assumptions of the lemma, $\hat{C}$ consists only of self-join-free queries. Thus, by Lemma~\ref{lem:self-join-free-counting-cores}, each query in $\hat{C}$ is its own $\ccore$ (up to deleting isolated variables). We will now consider the following cases:

\begin{enumerate}
\item[(i)] Suppose that $\hat{C}$ has unbounded degree of freedom. With Definition~\ref{def:contract} it is straightforward to check that a quantified variable $y$ with degree of freedom $B$ induces a clique of size $B$ in the contract of the corresponding query. Therefore, the contracts of the queries in $\hat{C}$ have unbounded treewidth. Consequently, $\#\homsprob(\hat{C})$ is $\ccW{1}$-hard by the classification of Chen and Mengel~\cite[Theorem 22]{ChenM15}. Since $\hat{C}\subseteq C$ the problem $\#\homsprob(\hat{C})$ is merely a restriction of $\#\ucq(C)$, the latter of which is thus $\ccW{1}$-hard as well.
\item[(ii)] Suppose that the degree of freedom of queries in $\hat{C}$ is bounded by a constant $d'$. We show that $\Gamma(C)$ has unbounded treewidth, which proves the claim by Theorem~\ref{thm:implicit_UCQ_dicho}. 
To this end, let $B$ be any positive integer. The goal is to find a conjunctive query in $\Gamma(C)$ with treewidth at least $B$. Since $\contract{\conjunctclass{C}}$ has unbounded treewidth, there is a UCQ $\Psi=(\mathcal{A}_1, \ldots,\mathcal{A}_\ell),X)$ in $C$ such that $\contract{\conjunctfull{\Psi}}$ has treewidth larger than $d + \binom{dd'}{2}+ B$. We will show that $\conjunctfull{\Psi}$ has treewidth larger than $d+B$, which, as we have argued previously, implies that $\Gamma(C)$ contains a query with treewidth at least $B$.
To prove that $\conjunctfull{\Psi}$ indeed has treewidth larger than $d+B$, let $\conjunctfull{\Psi}=(\mathcal{A},X)$ and let $G$ be the Gaifman graph of $\mathcal{A}$. Let $Y=U(\mathcal{A})\setminus X$ and recall 
from Definition~\ref{def:contract} that $\contract{\mathcal{A},X}$ is obtained from $G[X]$ by adding an edge between any pair of vertices $u$ and $v$ that are adjacent to a common connected component in $G[Y]$. Let $N\subseteq X$ be the set of all vertices in~$X$ that are adjacent to a vertex in $Y$ and observe that $|N|\leq dd'$ since the number of existentially quantified variables and the degree of freedom are bounded by $d$ and $d'$, respectively. Thus, $\contract{\mathcal{A},X}$ is obtained by adding at most $\binom{dd'}{2}$ edges to $G[X]$. The deletion of an edge can decrease the treewidth by at most $1$,  so
\[ \mathsf{tw}(\conjunctfull{\Psi}) = \mathsf{tw}(G) \geq \mathsf{tw}(G[X]) \geq \mathsf{tw}(\contract{\mathcal{A},X}) - \binom{dd'}{2} >  d + B  \,,\]
which concludes Case (ii).
\end{enumerate}
With all cases concluded, the proof is completed.
\end{proof}

From Lemmas~\ref{lem:UCQ_algo} and~\ref{lem:main_dichotomy_selfjoinfree} we directly obtain the main result of this section, which we restate for convenience.
\maindichogeneral*

\begin{remark}
    It turns out that all side conditions of Theorem~\ref{thm:main_dicho_quantifiers} are necessary if we aim to classify $\#\ucq(C)$ solely via $\conjunctclass{C}$. To this end, we provide counter examples for each missing condition in Appendix~\ref{sec:appendix}.
\end{remark}

\section{The Meta Complexity of Counting Answers to UCQs}\label{sec:meta}

We consider the meta-complexity question of deciding whether it is possible to count the answers of a given UCQ in linear time. 
As pointed out in the introduction, this problem is immediately $\mathrm{NP}$-hard even when restricted to conjunctive queries if we were to allow quantified variables. Therefore, we consider quantifier-free UCQs in this section. Recall the definition of $\meta$ from Section~\ref{sec:intro}.

\vbox{
\begin{description}\setlength{\itemsep}{0pt}
\setlength{\parskip}{0pt}
\setlength{\parsep}{0pt}   			
\item[\bf Name:] $\meta$ 
\item[\bf Input:]   A union $\Psi$ of quantifier-free conjunctive queries. 
\item[\bf Output:]  Is it possible to count answers to $\Psi$ in time linear in the size of $\calD$, i.e., can the function $\mathcal{D} \mapsto \ans{\Psi}{\mathcal{D}}$ be computed in time $O(|\mathcal{D}|)$. 
\end{description}
}

Since we focus in this section solely on quantifier-free queries, it will be  convenient to simplify our notation as follows. As all variables are free, we will identify a conjunctive query $\varphi$ just by its associated structure, that is, we will write $\varphi=\mathcal{A}$, rather than $\varphi=(\mathcal{A},U(\mathcal{A}))$. Similarly, we represent a union of quantifier-free conjunctive queries $\Psi$ as a tuple of structures $\Psi=(\mathcal{A}_1,\dots,\mathcal{A}_\ell)$.

For studying the complexity of $\meta$, it will be crucial to revisit the classification of linear-time counting of answers to quantifier-free conjunctive queries:
The following theorem is well known and was discovered multiple times by different authors in different contexts.\footnote{We remark that~\cite[Theorem 7]{BeraGLSS22} focuses on the special case of graphs and \emph{near} linear time algorithms. However, in the word RAM model with $O(\log n)$ bits, a linear time algorithm is possible~\cite{CarmeliZBCKS22}.}  
This is Theorem~\ref{thm:CQ_dichointro} from the introduction, which we now restate in a version that expresses CQs as structures.

\begin{theorem}[See Theorem 12 in~\cite{BraultBaron13}, and \cite{BeraGLSS22,Mengel25arxiv}]\label{thm:CQ_dicho}
Let $\mathcal{A}$ be a quantifier-free conjunctive query. If $\mathcal{A}$ is acyclic, then the function $\mathcal{D} \mapsto \homs{\mathcal{A}}{\mathcal{D}}$ is computable in linear time. Otherwise, if $\mathcal{A}$ has arity $2$, then $\mathcal{D} \mapsto \homs{\mathcal{A}}{\mathcal{D}}$ cannot be solved in linear time unless the Triangle Conjecture fails.
\end{theorem}

Theorem~\ref{thm:CQ_dicho} yields an efficient way to check whether counting answers to a quantifier-free conjunctive query $\varphi$ can be done in linear time: Just check whether $\varphi$ is acyclic. In stark contrast, we show that no easy criterion for linear time tractability of counting answers to \emph{unions} of conjunctive queries is possible, unless some conjectures of fine-grained complexity theory fail. In fact, our Theorem~\ref{thm:main_meta}, which we restate here for convenience, precisely determines the complexity of $\meta$.

\mainmeta*

The lower bounds in Theorem~\ref{thm:main_meta} imply that the exponential dependence on $\ell$ in our $2^{O(\ell)} \cdot |\Psi|^{\mathsf{poly}(\log|\Psi|)}$ time algorithm for $\meta$ cannot be significantly improved, unless standard assumptions fail.

The remainder of this section is devoted to the proof of Theorem~\ref{thm:main_meta}. It is split into two parts: In the first and easier part (Section~\ref{sec:meta_algo}), we construct the algorithm for $\meta$. For this, all we need to do is to translate the problem into the homomorphism basis (see Section~\ref{sec:UCQstoHoms}), i.e., we transform the problem of counting answers to $\Psi$ into the problem of evaluating a linear combination of terms, each of which can be determined by counting the answers to a conjunctive query. This is done in Lemma~\ref{lem:meta_algo}.
The second part (Section~\ref{sec:meta_hardness}) concerns the lower bounds and is more challenging. The overall strategy is as follows: In the first step, we use a parsimonious reduction 
from $3$-$\textsc{SAT}$ to computing the reduced Euler Characteristic of a complex.
The parsimonious reduction is due to Roune and S{\'{a}}enz{-}de{-}Cabez{\'{o}}n~\cite{RouneS13}. 
In combination with the Sparsification Lemma~\cite{ImpagliazzoPZ01}, this reduction becomes tight enough for our purposes. In the second step, we show how to encode a complex $\Delta$ into a union of acyclic conjunctive queries $\Psi$ such that the following is true: The reduced Euler Characteristic of $\Delta$ is zero if and only if all terms in the homomorphism basis are acyclic. The lower bound results of Theorem~\ref{thm:main_meta} are established in Lemmas~\ref{lem:meta_hard1}, and~\ref{lem:meta_hard3}.

\subsection{Solving $\meta$ via Inclusion-Exclusion}\label{sec:meta_algo}

\begin{lemma}\label{lem:new_meta_help}
    There is a deterministic algorithm that on input a quantifier-free UCQ $\Psi = (\mathcal{A}_1,\dots,\mathcal{A}_\ell)$ decides in time $|\Psi|^{\mathsf{poly}(\log|\Psi|)}\cdot 2^{O(\ell)}$ whether all structures $\mathcal{A}$ with
    $c_\Psi(\mathcal{A})\neq 0$ are acyclic.
\end{lemma}
\begin{proof}
We rely on Babai's algorithm for isomorphism of vertex coloured graphs: This result~\cite[Corollary 4]{Babai16} implies that the isomorphism problem for vertex-coloured $n$-vertex graphs can be solved in time $n^{\mathsf{poly}(\log n)}$. Given a structure $\mathcal{A}$, the incidence graph of $\mathcal{A}$ is a vertex-coloured bipartite graph $B(\mathcal{A})=(V_1,V_2,E)$ where $V_1 = U(\mathcal{A})$ and $V_2$ is the set of all tuples of $\mathcal{A}$. We colour each vertex in $V_1$ with colour $U$, and for each relation symbol $R$ of the signature of $\mathcal{A}$, we colour all tuples $\vec{t} \in R^{\mathcal{A}}$ with colour $R$. It is then easy to see that two structures are isomorphic if and only if their incidence graphs are isomorphic (as vertex-coloured graphs). Thus, Babai's result implies that we can decide isomorphism between two structures $\mathcal{A}$ and $\mathcal{A}'$ in time $m^{\mathsf{poly}(\log m)}$ where $m=\max\{|\mathcal{A}|,|\mathcal{A}'|\}$.

We are now able to complete the proof.
By Definition~\ref{def:coefUCQ}, \[ c_\Psi(\mathcal{A}) = \sum_{\substack{J\subseteq [\ell]\\ \conjuncts{\Psi}{J} \cong \mathcal{A}}} (-1)^{|J|+1}.\]
This suggests the following algorithm: 
for each subset $J\subseteq [\ell]$, compute $\conjuncts{\Psi}{J}$. Afterwards, using the aforementioned application of Babai's algorithm, collect the isomorphic terms and compute $c_\Psi(\conjuncts{\Psi}{J})$ for each $J\subseteq[\ell]$ in time $2^{O(\ell)} \cdot |\Psi|^{\mathsf{poly}(\log|\Psi|)}$ (observing that $|\conjuncts{\Psi}{J}|\in O(|\Psi|)$). Clearly, $c_\Psi(\mathcal{A})=0$ for every $\mathcal{A}$ that is not isomorphic to any $\conjuncts{\Psi}{J}$.

Finally, output $1$ if each $\mathcal{A}$ with $c_\Psi(\mathcal{A})\neq 0$ is acyclic (each of these checks can be done in linear time~\cite{TarjanY84}), and output $0$ otherwise. The total running time of this algorithm is bounded from above by $2^{O(\ell)} \cdot |\Psi|^{\mathsf{poly}(\log|\Psi|)}$.
\end{proof}

We are now able to provide our algorithm for $\meta$. As mentioned in the introduction, since the question asked by $\meta$ is itself a complexity-theoretic question, we need to rely on a lower bound assumption for the correctness of the algorithm. For the case of arity $2$, the Triangle Conjecture will suffice, but we remark that result below and its proof apply verbatim to the case of arbitrary arity (which is why we proved Lemma~\ref{lem:new_meta_help} for arbitrary arities); the only difference is that instead of assuming the Triangle Conjecture, the higher arity case relies on the Hyperclique Hypothesis (see~\cite{Mengel25,Mengel25arxiv}). Since we wish to avoid introducing yet another lower bound assumption, we only present the case of arity $2$.

\begin{lemma}\label{lem:meta_algo}
When restricted to queries of arity $2$, the problem $\meta$ can be solved in time $|\Psi|^{\mathsf{poly}(\log|\Psi|)}\cdot 2^{O(\ell)}$, where $\ell$ is the number of conjunctive queries in the union, if the Triangle Conjecture is true.  
\end{lemma}
\begin{proof}
If the Triangle Conjecture is true, then Theorem~\ref{thm:CQ_dicho} implies that counting answers to a quantifier-free conjunctive query of arity $2$ is solvable in linear time if and only if the query is acyclic.
In combination with Corollary~\ref{cor:UCQ_monotone} we obtain that counting answers to a UCQ $\Psi=(\mathcal{A}_1,\dots,\mathcal{A}_\ell)$ can be done in linear time if and only if each $\mathcal{A}$ with $c_\Psi(\mathcal{A})\neq 0$ is acyclic --- recall that we consider the quantifier-free case in which \#equivalence is the same as isomorphism, hence we do not need to consider \#cores.

Consequently, we only need to call the algorithm provided in Lemma~\ref{lem:new_meta_help}. This concludes the proof.
\end{proof}

\subsection{Fine-grained Lower bounds for $\meta$}\label{sec:meta_hardness}

For our hardness proof for $\meta$, we will construct a reduction from the computation of the reduced Euler characteristic of an abstract simplicial complex. We begin by introducing some central notions about (abstract) simplicial complexes.

\subsubsection{Simplicial Complexes}
A simplicial complex captures the geometric notion of an independence system. We will use the corresponding combinatorial description, which is also known as \emph{abstract simplicial complex}, and defined as follows.
\begin{definition}
A \emph{complex} (short for abstract 
simplicial complex) $\Delta$ is a pair consisting of a nonempty 
finite ground set $\Omega$ and a set of \emph{faces} $\mathcal{I}\subseteq 2^\Omega$ 
that includes all singletons and is a downset.  
That is, the set of faces~$\mathcal{I}$ satisfies the following criteria.
\begin{itemize}
    \item $\forall S\subseteq \Omega: S \in \mathcal{I} \Rightarrow \forall S'\subseteq S: S' \in \mathcal{I} $, and
    \item $\forall x \in \Omega: \{x\}\in \mathcal{I}$.
\end{itemize}
The inclusion-maximal faces in $\mathcal{I}$  are called \emph{facets}. Unless stated otherwise, we encode complexes  by the ground set $\Omega$ and the set of facets. Then $|\Delta|$ is its encoding
length.  
\end{definition}

\begin{definition}
The \emph{reduced Euler characteristic} of a complex~$\Delta=(\Omega,\mathcal{I})$ is defined as
\[ \hat{\chi}(\Delta) := - \sum_{S \in \mathcal{I}} (-1)^{|S|}  \,.\]
\end{definition}

Consider for example the complexes $\Delta_1$ and $\Delta_2$ in Figure~\ref{fig:Example_complexes}, given with a computation of their reduced Euler characteristic.

\begin{figure}
    \centering
    \includegraphics{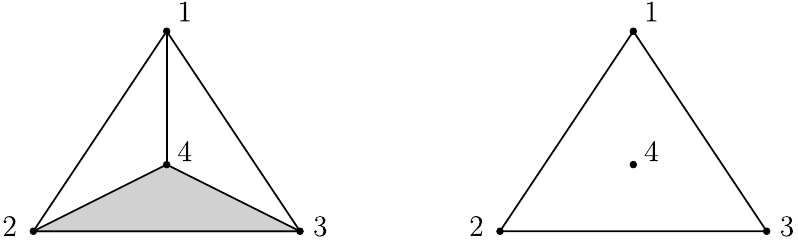}
    \caption{Two complexes over the groundset $\Omega=\{1,2,3,4\}$. Let $\Delta_1$ be the complex shown on the left. It has facets $\{2,3,4\}$, $\{1,2\}$, $\{1,3\}$, and $\{1,4\}$. Let $\Delta_2$ be the complex shown on the right, with facets $\{1,2\}$, $\{2,3\}$, $\{1,3\}$, and $\{4\}$. The reduced Euler characteristic of these complexes is computed as follows: Since $\Delta_1$ has one face of size $3$ ($\{2,3,4\}$), $6$ faces of size $2$, $4$ faces of size $1$, and the empty set as face of size $0$ it holds that $\hat{\chi}(\Delta_1) = - (-1 + 6 - 4 +1) = -2$. Similarly, we have $\hat{\chi}(\Delta_2) = - (3 - 4 + 1) = 0$.}
    \label{fig:Example_complexes}
\end{figure}

Let $\Delta=(\Omega,\mathcal{I})$ be a complex. 
Consider distinct elements~$x$ and~$y$ in~$\Omega$. We say that $x$ \emph{dominates} $y$ if, for every $S\in\mathcal{I}$, $y \in S$ implies $S\cup\{x\} \in \mathcal{I}$. 
For example, $x$ and $y$ dominate each other if they are contained in the same facets. We start with a simple observation:

\begin{lemma}\label{lem:dom_equivalence}
    Let $\Delta=(\Omega,\mathcal{I})$ be a complex and let $x,y\in \Omega$. Then $x$ dominates $y$ if and only if each facet that contains $y$ must also contain $x$.
\end{lemma}
\begin{proof}
    For the forward direction suppose $x$ dominates $y$. Let $F$ be a facet that contains $y$. Then $F\cup\{x\} \in \mathcal{I}$. Since facets are inclusion maximal in $\mathcal{I}$ it follows that $F=F\cup \{x\}$, that is, $x\in F$.

    For the backward direction suppose that each facet containing $y$ must also contain $x$. Let $S\in\mathcal{I}$ with $y \in S$. Then there is a facet $F$ with $S\subseteq F$. Since $x \in F$ we have $S \cup \{x\} \subseteq F$ and hence $S \cup \{x\} \in \mathcal{I}$. Therefore $x$ dominates $y$.
\end{proof}

We say that a complex $\Delta=(\Omega,\mathcal{I})$ is \emph{irreducible} if, for every $y\in \Omega$, there is no $x\in \Omega \setminus \{y\}$ that dominates~$y$. For example, both complexes in Figure~\ref{fig:Example_complexes} are irreducible.

Given $\Delta=(\Omega,\mathcal{I})$ and $y\in\Omega$, we define $\Delta\setminus y$ to be the complex obtained from $\Delta$ by deleting all faces that contain $y$ and by deleting $y$ from $\Omega$.
The following lemma seems to be folklore. We include a proof only for reasons of self-containment.
\begin{lemma}\label{lem:no_dom}
Let $\Delta = (\Omega,\mathcal{I})$ be a complex. If $y\in \Omega$ is dominated
by some $x\in \Omega\setminus \{y\}$ then
$\hat{\chi}(\Delta) = \hat{\chi}(\Delta\setminus y)$.
\end{lemma}
\begin{proof}
Let $\Delta=(\Omega,\mathcal{I})$. Write $\mathcal{I}_{y}$ for the set of all faces containing $y$. Consider the mapping $b:\mathcal{I}_y \to \mathcal{I}_y$
\[b(S) := \begin{cases} S \cup \{x\} & x \notin S \\ S \setminus\{x\} & x \in S\,.
\end{cases}
\]
Note that $b$ is well-defined since $S\cup \{x\} \in \mathcal{I}_y$ because $y\in S$ and $x$ dominates $y$.
Observe that $b$ induces a partition of $\mathcal{I}_y$ in pairs $\{S,S\cup\{x\}\}$ for $x\notin S$. For those pairs, we clearly have $|S|+1 =|S\cup\{x\}|$. Thus
\[ \sum_{S \in \mathcal{I}_y} (-1)^{|S|} = 0 \,,\]
and therefore
\[ \hat{\chi}(\Delta) = - \sum_{S \in \mathcal{I}} (-1)^{|S|} = - \sum_{S \in \mathcal{I}\setminus \mathcal{I}_y} (-1)^{|S|} = \hat{\chi}(\Delta\setminus y)\,. \]
\end{proof}

\begin{definition}
Two complexes $\Delta_1=(\Omega_1,\mathcal{I}_1)$ and $\Delta_2=(\Omega_2,\mathcal{I}_2)$ are \emph{isomorphic} if there is a bijection $b:\Omega_1 \to \Omega_2$ such that $S \in \mathcal{I}_1$ if and only if $b(S)\in \mathcal{I}_2$ for each $S\subseteq \Omega$, where $b(S)=\{b(x)\mid x \in S\}$.
\end{definition}

Finally, a complex $(\Omega,\mathcal{I})$ is called \emph{trivial} if it is isomorphic to $(\{x\},\{\emptyset,\{x\}\})$.

\subsubsection{The Main Reduction}\label{sec:mainred}

To begin with, we require a conjunctive query whose answers cannot be counted in linear time under standard assumptions. To this end, we define, for positive integers $k$ and $t$, a binary relational structure $\mathcal{K}_t^k$ as follows. Start with a $t$-clique $K_t$ and $k$-stretch every edge, that is, each edge of $K_t$ is replaced by a path consisting of $k$ edges. We denote the resulting graph by~$K_t^k$. For each edge $e$ of $K_t^k$, we introduce a relation $R_e=\{e\}$ of arity $2$. The structure $\mathcal{K}_t^k$ has universe $V(K_t^k)$ and relations $(R_e)_{e\in E(K_t^k)}$.

\begin{observation}\label{fac:obvious_facts_are_obvious}
Let $k$ and $t$ be positive integers. The structure $\mathcal{K}_t^k$ is self-join-free and has arity $2$.
\end{observation}

\begin{lemma}\label{lem:non_inif_ETH_clique}
Suppose that the non-uniform ETH holds. For all positive integers $d$, there is a $t$ such that for each $k$, the function  $\calD\mapsto \homs{\mathcal{K}_t^k}{\calD}$ cannot be computed in time $O(\abs{\calD}^d)$.
\end{lemma}
\begin{proof}
We  again write $K_t$ for the $t$-clique and $K_t^k$ for the $k$-stretch of $K_t$.
Chen, Eickmeyer and Flum~\cite{ChenEF12} proved that, under non-uniform ETH, for each positive integer $d$, there is a $t$ such that determining whether a graph $G$ contains $K_t$ as a subgraph cannot be done in time $O(|G|^d)$.

We construct a simple linear-time reduction to computing $\calD\mapsto \homs{\mathcal{K}_t^k}{\calD}$. Since for each edge $e$ of the underlying graph $K_t^k$, the structure $\mathcal{K}_t^k$ has a separate binary single-element relation $R_e$, the input database $\calD$ (of the same signature as $\mathcal{K}_t^k$) also has a relation $R'_e$ whose elements can be interpreted as $e$-coloured edges. This means that the problem of computing $\calD\mapsto \homs{\mathcal{K}_t^k}{\calD}$ can equivalently be expressed as the following problem: Given a graph $G'$, that comes with an edge-colouring $\mathsf{col}: E(G') \mapsto E(K_t^k)$, count the homomorphisms $h$ from $K_t^k$ to $G'$ such that for each edge $e$ of $K_t^k$ we have $\mathsf{col}(h(e))=e$.

We now reduce the problem of determining whether a graph $G$ contains $K_t$ to this problem.    
Given input $G$, we construct an edge-coloured graph $G'$ as follows. Each edge of $G$ is replaced by $\binom{t}{2}$ paths of length $k$, and we colour the edges of the $i$-th of those paths with the $k$ edges of the $k$-stretch of the $i$-th edge $e_i$ of $K_t$. It is easy to see that $G$ contains a $t$-clique if and only if there is at least one homomorphism from $K_t^k$ to $G'$ that preserves the edge-colours. Moreover, the construction of $G'$ can clearly be done in linear time.
\end{proof}

Before proving Theorem~\ref{thm:main_meta}, we need to take a detour to examine the complexity of computing the reduced Euler Characteristic of a complex associated with a UCQ. To begin with, we introduce the notion of a ``power complex''.

\begin{definition}
Let $\mathcal{U}$ be a finite set, and let $\Omega \subseteq 2^\mathcal{U}$ be a set system that does not contain $\mathcal{U}$. The \emph{power complex} $\Delta_{\Omega,\mathcal{U}}$ is a complex with ground set $\Omega$ and faces 
\[\mathcal{I}=\left\{ S \subseteq \Omega ~:~ \bigcup_{A\in S} A \neq \mathcal{U} \right\}\,.\] 
\end{definition}

It is easy to check that $\Delta_{\Omega,\mathcal{U}}$ is a complex --- for each $x\in \Omega$, the set $\{x\}$ is in $\mathcal{I}$ since $\Omega$ does not contain $\mathcal{U}$. Although we typically encode complexes by the ground set $\Omega$ and the set of facets, in the case of a power complex 
we list the elements of $\mathcal{U}$ and $\Omega$.

In the first step, we show that each complex is isomorphic to a power complex.

\begin{lemma}\label{lem:I_have_the_power!}
Let $\Delta=(\Omega,\mathcal{I})$ be a non-trivial irreducible complex. It is possible to
compute, in polynomial time in $|\Delta|$, a set $\mathcal{U}$ and a set $\hat{\Omega}\subseteq 2^\mathcal{U}$ with $\mathcal{U}\notin \hat{\Omega}$ such that $\Delta$ is isomorphic to the power complex $\Delta_{\hat{\Omega},\mathcal{U}}$.
\end{lemma}
\begin{proof}
Let $F_1,\dots,F_k$ be the facets of $\Delta$ so that the encoding of~$\Delta$ consists of $\Omega$ and $F_1,\ldots,F_k$, and $|\Delta|$ is the length of this encoding. 
For each $i\in [k]$, we introduce an element $E_i$ corresponding to $F_i$.
Let  $\mathcal{U}=\{E_1,\dots,E_k\}$. Next, we define a mapping $b:\Omega \to 2^\mathcal{U}$ as follows:
\[ b(x) \coloneqq \{E_i \mid x \notin F_i \} \,.\]
Observe that $b$ is injective since otherwise two elements of $\Omega$ are contained in the same facets of $\Delta$, which means that they dominate each other. Also, note that $b(x)=\mathcal{U}$ would imply that $x$ is not contained in any facet of $\Delta$ which gives a contradiction as $\{x\}\in \mathcal{I}$ for any $x \in \Omega$.
    
We choose $\hat{\Omega}$ as the range of $b$. Then, clearly, $b$ is a bijection from $\Omega$ to $\hat{\Omega}$. Furthermore, $\mathcal{U}$ and $\hat{\Omega}$ can be constructed in time polynomial in $|\Delta|$.

It remains to prove that $b$ is also an isomorphism from $\Delta$ to $\Delta_{\hat{\Omega},\mathcal{U}}$. 
To this end, we need to show that \[S\in \mathcal{I} \Leftrightarrow \bigcup_{A \in b(S)} A \neq \mathcal{U}\,,\]
where $b(S)=\{b(x)\mid x \in S\}$.

For the first direction, let $S\in\mathcal{I}$. W.l.o.g.\ we have $S\subseteq F_1$. Then, for all $x\in S$, $E_1\notin b(x)$. Hence $E_1 \notin \bigcup_{A \in b(S)} A$. For the other direction, suppose that $S\notin \mathcal{I}$. Then $S$ is not a subset of any facet. Consequently, for every $i\in[k]$, there exists $x_i \in S$ such that $x_i \notin F_i$. By definition of $b$, we thus have that, for every $i\in[k]$, there exists $x_i \in S$ such that $E_i \in b(x_i)$. Thus
    $\bigcup_{A \in b(S)} A = \mathcal{U}$.
\end{proof}

\begin{example}\label{example1}
Recall the complex $\Delta_1$ in Figure~\ref{fig:Example_complexes}. $\Delta_1$ has facets $F_1=\{2,3,4\}$, $F_2=\{1,2\}$, $F_3=\{1,3\}$, and $F_4=\{1,4\}$. Since $\Delta_1$ is irreducible and non-trivial, we can apply the construction in the previous lemma to construct an isomorphic power complex: We set $\mathcal{U}=\{E_1,E_2,E_3,E_4\}$, since $\Delta_1$ has $4$ facets. Moreover, the ground set $\hat{\Omega}$ of the power complex is the range of the function $b:\Omega \to 2^\mathcal{U}$ with $b(x) \coloneqq \{E_i \mid x \notin F_i \}$,
that is, $b(1)=\{E_1\}$ since $1\notin F_1$, and similarly, $b(2)=\{E_3,E_4\}$, $b(3)=\{E_2,E_4\}$, and $b(4)=\{E_2,E_3\}$. 
Then the facets $\hat{\mathcal{I}}$ of the power complex are $\{E_1,E_3,E_4\}$, $\{E_1,E_2,E_4\}$, $\{E_1,E_2,E_3\}$, and $\{E_2,E_3,E_4\}$.
Then $\Delta_1$ and the power complex $\Delta_{\hat{\Omega},\mathcal{U}}$ are isomorphic, where for instance $F_2=\{1,2\}$ corresponds to $b(1)\cup b(2) =\{E_1,E_3,E_4\}$.
When applied to $\Delta_2$, the same construction yields a power complex with $\mathcal{U}=\{E_1,E_2,E_3,E_4\}$ and ground set $\{\{E_3,E_4\},\{E_1,E_4\},\{E_2,E_4\},\{E_1,E_2,E_3\}\}$.
\end{example}

We now introduce the main technical result that we use to establish lower bounds for $\meta$; note that computability result in the last part of Lemma~\ref{lem:main_reduct_meta_help} will just be required for the construction of exemplary classes of UCQs in the appendix.

\begin{lemma}\label{lem:main_reduct_meta_help}
    For each positive integer $t$, there is a polynomial-time algorithm $\hat{\mathbb{A}}_t$ that, when given as input a non-trivial irreducible complex $\Delta=(\Omega,\mathcal{I})$ with $\Omega\notin \mathcal{I}$,  computes a union of quantifier-free conjunctive queries $\Psi=(\mathcal{B}_1,\dots,\mathcal{B}_\ell)$ satisfying the following constraints:
    \begin{enumerate}
        \item $\conjunctfull{\Psi} \cong \mathcal{K}_t^k$ for some $k\geq 1$.
        \item $c_\Psi(\conjunctfull{\Psi}) = - \hat{\chi}(\Delta)$. 
        \item For all relational structures $\mathcal{B}\ncong \conjunctfull{\Psi}$, $c_\Psi(\mathcal{B})\neq 0$ implies that $\mathcal{B}$ is acyclic.
        \item $\ell\leq|\Omega|$.
        \item For all $i\in [\ell]$ the conjunctive query $\mathcal{B}_i$ is acyclic and self-join-free, and has arity $2$.
    \end{enumerate}
    Moreover, the algorithm $\hat{\mathbb{A}}_t$ can be explicitly constructed from~$t$.
\end{lemma}

\begin{proof}
The algorithm $\hat{\mathbb{A}}_t$ applies Lemma~\ref{lem:I_have_the_power!} to obtain $\mathcal{U}$ and $\hat{\Omega}$ such that the power-complex $\Delta_{\hat{\Omega},\mathcal{U}}$ is isomorphic to $\Delta$. Let $k=|\mathcal{U}|$ and assume w.l.o.g.\ that $\mathcal{U}=[k]$. Let $\hat{\Omega}=\{A_1,\dots,A_\ell\}$. Let $e^j$ be the $j$-th edge of the $k$-stretch of $e$ in $K_t^k$ and recall that $\mathcal{K}_t^k$ contains the relations $R_{e^i}=\{e^i\}$ for each $e\in E(K_t)$ and $i\in[k]$. For each $i\in[k]$ we define $\mathcal{E}_i$ to be the substructure of $\mathcal{K}_t^k$ containing 
the universe $V(K_t^k)$ and
the relations $R_{e^i}$ (one relation for each $e\in E(K_t)$). The algorithm constructs $\Psi=(\mathcal{B}_1,\dots,\mathcal{B}_\ell)$ as follows:
For each $j\in[\ell]$, set $\mathcal{B}_j=\cup_{i\in A_j} \mathcal{E}_i$, that is, $\mathcal{B}_j$ is the substructure of $\mathcal{K}_t^k$ obtained by taking all relations in $\mathcal{E}_i$ for all $i\in A_j$.  The pseudocode for $\hat{\mathbb{A}}_t$ is as follows.
\begin{algorithm}
\caption{Pseudocode for $\hat{\mathbb{A}}_t$}\label{alg:cap}
\begin{algorithmic}[1]
\State \textbf{Input:} $\Delta=(\Omega,\mathcal{I})$
\State $(F_1,\dots,F_k) \gets $ facets of $\Delta$ 
    \For{$x\in \Omega$} \State $A_x \gets \{i \mid x \notin F_i \}$ \Comment{$\mathcal{U}=[k]$ and $\hat{\Omega}= \{A_x \mid x \in \Omega\} = \{A_1,\ldots,A_{\ell}\}$ yield the power   complex}
    \EndFor
\State $K_t \gets$ $t$-clique, $m \gets \binom{t}{2}$ \Comment{$e_i$ denotes the $i$-th edge of $K_t$}
\State $K^k_t \gets $ $k$-stretch of $K_t$ \Comment{$e_i^j$ denotes the $j$-th edge of the $k$-stretch of $e_i\in E(K_t)$}
\For{$i\in [m]$, $j\in[k]$} \State $R_{e_i^j} \gets \{e_i^j \}$
    \EndFor
    \For{$i\in[k]$}
\State $\mathcal{E}_i \gets (V(K_k^t),R_{e_1^i},\dots,R_{e_m^i})$
\EndFor
\For{$j\in [\ell]$}
\State $\mathcal{B}_j\gets \bigcup_{i\in A_j} \mathcal{E}_i$
\EndFor
\State \textbf{Output:} $\Psi=(\mathcal{B}_1,\dots,\mathcal{B}_\ell)$
\end{algorithmic}
\end{algorithm}

    We now prove the required properties of $\Psi$.
    \begin{enumerate}
\item Recall that $\Omega$ is not a facet. Since $\Delta$ and $\Delta_{\hat{\Omega},\mathcal{U}}$ are isomorphic, $\hat{\Omega}$ is not a facet. Thus $\bigcup_{i\in [\ell]} A_i = \mathcal{U}$ by the definition of power complexes. Since $\mathcal{B}_j$ contains all relations in $\mathcal{E}_i$ with $i\in A_j$, we have that $\conjunctfull{\Psi}= \mathcal{K}_t^k$ as desired.
\item  Recall from Definition~\ref{def:coefUCQ} that $\mathcal{I}(\mathcal{B},X)=\{J \subseteq [\ell] \mid (\mathcal{B},X) \sim \conjuncts{\Psi}{J}\}$. Since we are in the setting of quantifier-free queries ($X=U(\mathcal{B})$; we will drop the $X$ as before), equivalence $(\sim)$ is equivalent to isomorphism. Thus, we have that $c_\Psi(\conjunctfull{\Psi})$ is equal to
 \[  
 \sum_{\substack{J\subseteq [\ell]\\ \conjunctfull{\Psi} \cong \conjuncts{\Psi}{J}}} (-1)^{|J|+1} = 
 \sum_{\substack{J\subseteq [\ell]\\ \conjunctfull{\Psi} \not\cong \conjuncts{\Psi}{J}}} (-1)^{|J|} = 
 \sum_{\substack{J\subseteq [\ell]\\ \cup_{j\in J}A_j \neq \mathcal{U}}} (-1)^{|J|} = -\hat{\chi}(\Delta_{\hat{\Omega},\mathcal{U}}) = -\hat{\chi}(\Delta) \,,\]
 where the first equality follows from the formula $\sum_{J \subseteq [\ell]} (-1)^{|J|} =0$.

\item If $\mathcal{B}\ncong \conjunctfull{\Psi}$ and $c_\Psi(\mathcal{B})\neq 0$, then $\mathcal{B}$ is a substructure of $\mathcal{K}_t^k$ that is missing at least one of the $\mathcal{E}_i$. The claim holds since each $\mathcal{E}_i$ is a feedback edge set (i.e., the deletion of all tuples in $\mathcal{E}_i$ breaks every cycle), because deleting $\mathcal{E}_i$ corresponds to deleting one edge from each stretch. 
\item Follows immediately by construction.
\item The $\mathcal{B}_i$ are self-join-free and of arity $2$ since they are substructures of $\mathcal{K}_t^k$ (see Observation~\ref{fac:obvious_facts_are_obvious}). For acyclicity, we use the fact that for each $j$ we have $A_j\neq\mathcal{U}$ by definition of the power complex, which implies the claim by the same argument as 3.
    \end{enumerate}
    With all cases completed, the proof is concluded.
\end{proof}

To provide a concrete example,  we apply the construction in 
Lemma~\ref{lem:main_reduct_meta_help} to the complexes $\Delta_1$ and $\Delta_2$ in Figure~\ref{fig:Example_complexes}. For this example, we will choose $t=3$. Since both $\Delta_1$ and $\Delta_2$ have four facets, we will choose $k=4$. The structure $\mathcal{K}_t^k$ is depicted in Figure~\ref{fig:samplequeries}, together with six selected substructures corresponding to the $\mathcal{B}_j$ in the proof of the Lemma~\ref{lem:main_reduct_meta_help}.
\begin{figure}
\centering
\includegraphics[width=\textwidth]{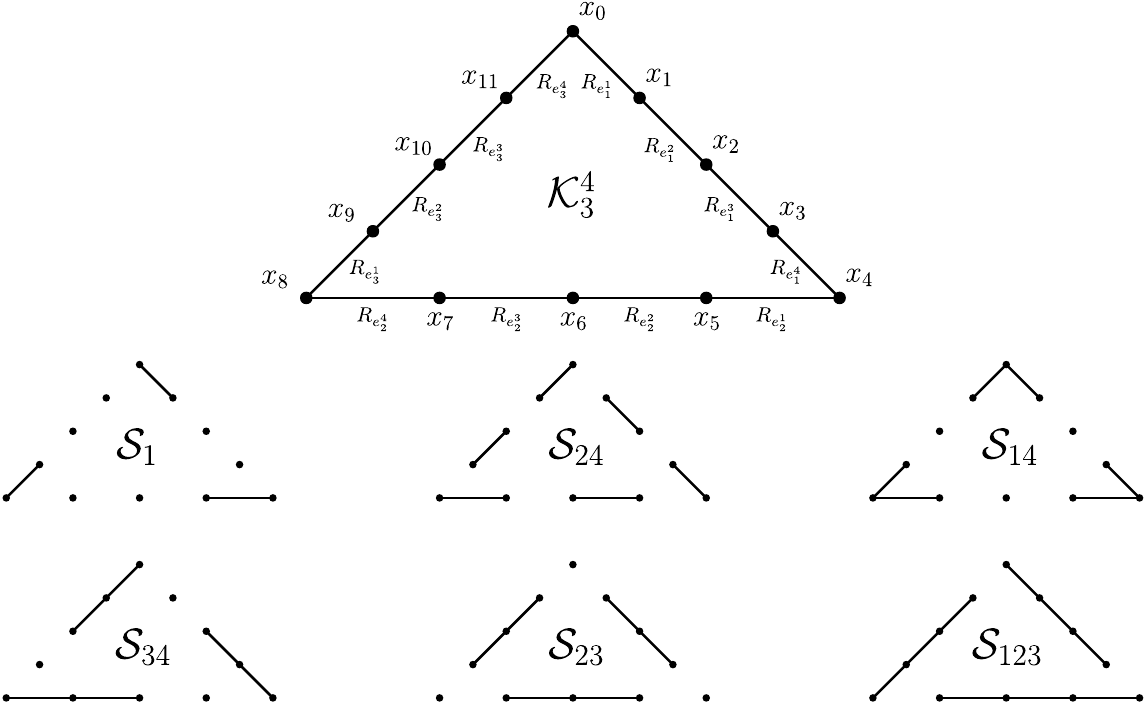}
\caption{\emph{(Top:)} The structure $\mathcal{K}_3^4$. \emph{(Bottom:)} Substructures $\mathcal{S}_A$ for some selected $A\subseteq [4]$. Observe that all of the $\mathcal{S}_A$ are \emph{acyclic}.}
    \label{fig:samplequeries}
\end{figure}

Let us start by applying $\hat{\mathbb{A}}_3$ to $\Delta_1$. Recall from Example~\ref{example1} that the power complex of $\Delta_1$ has ground set $\{A_1,A_2,A_3,A_4\}$ with $A_1=\{E_1\}$, $A_2=\{E_3,E_4\}$, $A_3=\{E_2,E_4\}$, and $A_4=\{E_2,E_3\}$. Thus $\hat{\mathbb{A}}_3$ returns the UCQ $\Psi_1=(\mathcal{S}_1,\mathcal{S}_{34}, \mathcal{S}_{24}, \mathcal{S}_{23})$. 
Similarly, applying $\hat{\mathbb{A}}_3$ to $\Delta_2$ yields the UCQ $\Psi_2=(\mathcal{S}_{24},\mathcal{S}_{34}, \mathcal{S}_{14}, \mathcal{S}_{123})$.

For the purpose of illustration, let us write $\Psi_1$ and $\Psi_2$ as formulas. To this end, given a non-empty subset $A$ of $\{1,2,3,4\}$, define the conjunctive query $\varphi_A$ as follows:
\[ \varphi_A = \bigwedge_{a\in A}  R_{e_1^a}(x_{a-1},x_{a}) \wedge R_{e_2^a}(x_{4+a-1},x_{4+a}) \wedge R_{e_3^a}(x_{8+a-1},x_{8+a}) \,,\]
where indices are taken modulo $12$. Observe that the $\mathcal{S}_A$ depicted in Figure~\ref{fig:samplequeries} are the structures associated to $\varphi_A$ for $A \in \{\{1\},\{2,4\},\{1,4\},\{3,4\},\{2,3\},\{1,2,3\}\}$. Then
\begin{align*}
    \Psi_1(x_0,\dots,x_{11}) &= \varphi_{1} \vee \varphi_{34} \vee \varphi_{24} \vee \varphi_{23}\,, \text{ and}\\
    \Psi_2(x_0,\dots,x_{11}) &= \varphi_{24} \vee \varphi_{34} \vee \varphi_{14} \vee \varphi_{123}
\end{align*}
Note that $\conjunctfull{\Psi_1}=\conjunctfull{\Psi_2}=\mathcal{K}^4_3$. Now recall that $\hat{\chi}(\Delta_1)\neq 0$ and $\hat{\chi(\Delta_2)}=0$. Thus, by Item 2.\ of Lemma~\ref{lem:main_reduct_meta_help}, it holds that $c_{\Psi_1}(\mathcal{K}^4_3) \neq 0$, which means that there is an acyclic CQ in the CQ expansion of $\Psi_1$. On the contrary, by Item 3.\ of Lemma~\ref{lem:main_reduct_meta_help}, for all $(\mathcal{A},X)$, we have that $c_{\Psi_2}(\mathcal{A},X)\neq 0$ implies that $(\mathcal{A},X)$ is acyclic. So, as a conclusion of our example, using complexity monotonicity (Corollary~\ref{cor:UCQ_monotone}) and Theorem~\ref{thm:CQ_dicho}, we obtain as an immediate consequence:

\begin{corollary}\label{cor:examples}
While it is not possible to count answers to $\Psi_1$ in linear time (in the input database), unless the Triangle Conjecture fails, for $\Psi_2$ such a linear-time algorithm exists (even though $\conjunctfull{\Psi_1}=\conjunctfull{\Psi_2}$).
\end{corollary}
\begin{proof}
By Corollary~\ref{cor:UCQ_monotone}, it is possible to count answers to a UCQ $\Psi$ in linear time if and only if for each \cminimal~conjunctive query $(\mathcal{A},X)$ with $c_\Psi(\mathcal{A},X)\neq 0$, it is possible to count answers to $(\mathcal{A},X)$ in linear time. Since $\Psi_1$ and $\Psi_2$ are quantifier-free, all conjunctive queries $\conjuncts{\Psi_1}{J}$ and $\conjuncts{\Psi_2}{J}$ are also quantifier-free and thus $\cminimal$. For all $(\mathcal{A},X)$, $c_{\Psi_2}(\mathcal{A},X)\neq 0$ implies that $(\mathcal{A},X)$ is acyclic, so it follows  from Theorem~\ref{thm:CQ_dicho} that it is possible to count answers to $\Psi_2$ in linear time.
Moreover, given that $\mathcal{K}^4_3$ is not acyclic, Theorem~\ref{thm:CQ_dicho} also implies that counting answers to $\mathcal{K}^4_3$ cannot be done in linear time, unless the Triangle Conjecture fails. Since $c_{\Psi_1}(\mathcal{K}^4_3) \neq 0$, 
  counting answers to $\Psi_1$ cannot be done in linear time, unless the Triangle Conjecture fails.\end{proof}

We will now conclude the hardness proof for $\meta$.
\begin{lemma}\label{lem:main_reduct_meta}
    For each positive integer $t$, there is a polynomial-time algorithm $\mathbb{A}_t$ that, when given as input a complex $\Delta=(\Omega,\mathcal{I})$, either computes $\hat{\chi}(\Delta)$, or computes a union of quantifier-free conjunctive queries $\Psi=(\mathcal{B}_1,\dots,\mathcal{B}_\ell)$ satisfying the following constraints:
    \begin{enumerate}
        \item $\conjunctfull{\Psi} \cong \mathcal{K}_t^k$ for some $k\geq 1$.
        \item $c_\Psi(\conjunctfull{\Psi}) = - \hat{\chi}(\Delta)$. 
        \item For all relational structures $\mathcal{B}\ncong \conjunctfull{\Psi}$, $c_\Psi(\mathcal{B})\neq 0$ implies that $\mathcal{B}$ is acyclic.
        \item $\ell\leq|\Omega|$.
        \item For all $i\in [\ell]$ the conjunctive query $\mathcal{B}_i$ is acyclic and self-join-free, and has arity $2$.
    \end{enumerate}
\end{lemma}

 \begin{proof}
Given $\Delta=(\Omega,\mathcal{I})$, we can successively apply Lemma~\ref{lem:no_dom} without changing the reduced Euler characteristic, until 
the resulting simplicial complex is irreducible.  This can be done in polynomial time:  By Lemma~\ref{lem:dom_equivalence} it suffices to check whether there are $y\neq x \in \Omega$ such that each facet containing $y$ also contains $x$. If no such pair exists, we are done. Otherwise we delete $y$ from $\Omega$ and from all facets and continue recursively. Clearly, the number of recursive steps is bounded by $|\Omega|$ so the run time is at most a polynomial in $|\Delta|$.

If this process makes the complex trivial, we output $0$ (i.e., the reduced Euler Characteristic of the trivial complex). We can furthermore assume that $\Omega$ is not a facet, i.e., that $\Omega \notin \mathcal{I}$, since in this case every subset of $\Omega$ is a face and the reduced Euler Characteristic is $0$. We can thus assume that $\Delta$ is non-trivial and irreducible, and that $\Omega$ is not a facet. Therefore, we can use the algorithm $\hat{\mathbb{A}}_t$ from Lemma~\ref{lem:main_reduct_meta_help}. This concludes the proof.
\end{proof}

We are now able to prove our lower bounds for $\meta$. 

\begin{lemma}\label{lem:meta_hard1}
If the Triangle Conjecture is true then $\meta$ is $\mathrm{NP}$-hard. If the Triangle Conjecture and ETH are both true  then $\meta$ cannot be solved in time $2^{o(\ell)}$ where $\ell$ is the number of conjunctive queries in its input.
Both results remain true
even if the input to $\meta$ is restricted to be over a binary signature.
\end{lemma}
\begin{proof}

For the first result, we  assume the Triangle Conjecture and show that $\meta$ is NP-hard.
The input to $\meta$ is 
a formula $\Psi'$ which is
a union of quantifier-free, self-join-free, and acyclic conjunctive queries. The goal is to decide whether 
counting answers of~$\Psi'$ (in an input database) can be done in linear time.

We reduce from $3$-SAT.
Let $F$ be a $3$-SAT formula.
The first step of our reduction is to apply a reduction from~\cite{RouneS13}. 
Concretely, \cite{RouneS13} gives a reduction   that, given a $3$-SAT formula $F$ with $n$ variables and $m$ clauses, outputs in polynomial time a complex $\Delta$ such that $F$ is satisfiable if and only if $\hat{\chi}(\Delta)\neq 0$. Moreover, the ground set of $\Delta$ has size $O(n+m)$.

Let $t=3$ and let $\mathbb{A}_t$ be the algorithm from Lemma~\ref{lem:main_reduct_meta}. Consider running $\mathbb{A}_t$ with input~$\Delta$.
If $\mathbb{A}_t$ outputs $\hat{\chi}(\Delta)$ then we can check immediately whether $\hat{\chi}(\Delta)=0$, which determines whether or not $F$ is satisfiable.
Otherwise, $\mathbb{A}_t$ outputs a formula~$\Psi$ which is a
union of self-join-free, quantifier-free, and  acyclic conjunctive queries of arity $2$, and further has the property that $c_{\Psi}(\conjunctfull{\Psi}) = - \hat{\chi}(\Delta)$. 
Since $\mathbb{A}_t$ is a polynomial-time algorithm, the number of conjunctive queries $\ell$ in $\Psi$ is at most a polynomial in $n+m$.

We wish to show that 
determining whether counting answers of~$\Psi$ can be done in linear time
would reveal whether or not 
$c_{\Psi} (\conjunctfull{\Psi})=0$ (which would in turn reveal whether $F$ is satisfiable).

By Corollary~\ref{cor:UCQ_monotone}, it is possible to compute $\calD\mapsto \ans{\Psi}{\calD}$ in linear time if and only if, for each relational structure $\mathcal{A}$ with $c_\Psi(\mathcal{A})\neq 0$, the function $\calD\mapsto \homs{\mathcal{A}}{\calD}$ can be computed in linear time.

So the intermediate problem is to check whether, for each relational structure   $\mathcal{A}$ with $c_\Psi(\mathcal{A})\neq 0$, the function $\calD\mapsto\homs{\mathcal{A}}{\calD}$ can be computed in linear time.
We wish to show that solving  the intermediate problem (in polynomial time) would enable us to determine whether or not
$c_{\Psi} (\conjunctfull{\Psi})=0$ (also in polynomial time).

Item~1 of Lemma~\ref{lem:main_reduct_meta} implies that
there is a positive integer~$k$ such that
$\conjunctfull{\Psi}\cong \mathcal{K}_3^k$.
Thus, $\conjunctfull{\Psi}$ is not acyclic.
However,
Item~3 of Lemma~\ref{lem:main_reduct_meta} implies that for every relational structure 
$\mathcal{A}\ncong \conjunctfull{\Psi}$
in the intermediate problem,
the structure $\mathcal{A}$ is acyclic. 
Theorem~\ref{thm:CQ_dicho} shows (assuming the Triangle Conjecture) that, for each $\mathcal{A}$,
$\calD\mapsto \homs{\mathcal{A}}{\calD}$ can be computed in linear time if and only if $\mathcal{A}$ is acyclic.
We conclude that the answer to the intermediate problem is yes iff 
$c_\Psi(\conjunctfull{\Psi})= 0$, completing the proof that $\meta$ is NP-hard.

To obtain the second result, we 
assume both the Triangle Conjecture and ETH. In this case, we
apply the 
Sparsification Lemma~\cite{ImpagliazzoPZ01} to the initial $3$-SAT formula $F$ before invoking the reduction.
By the Sparsification Lemma, it is possible in time $2^{o(n)}$ to construct $2^{o(n)}$ formulas $F_1,F_2,\dots$ such that $F$ is satisfiable if and only if at least one of the $F_i$ is satisfiable. Additionally, each $F_i$ has $O(n)$ clauses. As before, for each such $F_i$ we obtain in polynomial time a corresponding complex~$\Delta_i$, whose ground set has size $O(n)$. For each $\Delta_i$,  the algorithm from Lemma~\ref{lem:main_reduct_meta} either outputs its reduced Euler characteristic 
$\hat{\chi}(\Delta_i)$ 
or outputs a  UCQ $\Psi_i$
which has the property that 
that $c_{\Psi}(\conjunctfull{\Psi_i}) = - \hat{\chi}(\Delta_i)$. 

If, for any~$i$, the algorithm outputs a value $\hat{\chi}(\Delta_i)\neq 0$
then $F_i$ is satisfiable, so $F$ is satisfiable.
Let $I$ be the set of indices~$i$ such that the algorithm outputs a UCQ $\Psi_i$.
The argument from the first result shows that $F$ is satisfiable if and only if there is an $i\in I$ such that counting answers to $\Psi_i$ can be done in linear time.

We will argue that a $2^{o(\ell)}$ algorithm for $\meta$ (where $\ell$ is the number of CQs in its input) would make it possible to determine in  $2^{o(n)}$ time whether 
there is an $i\in I$ such that counting answers to $\Psi_i$ can be done in linear time.
This means that a $2^{o(\ell)}$ algorithm for $\meta$ would make it possible to determine in  $2^{o(n)}$ time whether $F$ is satisfiable, contrary to ETH.

To do this, we just need to show that the number of CQs in $\Psi_i$, which we denote $\ell(\Psi_i)$, is $O(n)$.
This follows since the ground set of $\Delta_i$ has size $O(n)$ and by Item~4 of Lemma~\ref{lem:main_reduct_meta},   $\ell(\Psi_i)$ is at most the size of the ground set.
Since the size of $I$ is $2^{o(n)}$ 
the running time for determining whether $F$ is satisfiable is $2^{o(n)}$ (for the sparsification)
plus $|I| \,\mathrm{poly}(n) = 2^{o(n)}$ time (for computing the complexes $\Delta_i$) plus $ \sum_{i\in I} 2^{o(\ell(\Psi_i))} = 2^{o(n)}$ for the calls to $\meta$, contradicting ETH, as desired.
\end{proof}

The proof of the following lemma is analogous, with the only exception that we do not invoke Theorem~\ref{thm:CQ_dicho} but apply Lemma~\ref{lem:non_inif_ETH_clique} with $d=1$ to obtain a $t$ such that for each $k$ the function $\calD\mapsto\homs{\mathcal{K}_t^k}{\calD}$ cannot be evaluated in linear time.

\begin{lemma}\label{lem:meta_hard3}
    If non-uniform ETH and the Triangle Conjecture are true then $\meta$ is $\mathrm{NP}$-hard and, furthermore, 
    \[\meta \notin \bigcap_{\varepsilon>0} \mathrm{DTime}(2^{\varepsilon \cdot \ell})\,.\]
    This remains true even if the input to $\meta$ is restricted to be over a binary signature.
\end{lemma}

Theorem~\ref{thm:main_meta} now follows immediately from Lemmas~\ref{lem:meta_algo},~\ref{lem:meta_hard1}, and~\ref{lem:meta_hard3}. 

Finally, we point out that our construction shows, in fact, something much stronger than just the intractability of deciding whether we can count answers to a UCQ in linear time: For any pair $(c,d)$ of positive integers 
satisfying $c \leq d$, it is hard to distinguish whether counting answers to a given UCQ can be done in time $O(n^c)$, or whether it takes time at least $\omega(n^d)$. Formally, we introduce the following gap problem:
\begin{definition}
Let $c$ and $d$ be positive integers with $c \leq d$. The problem $\meta[c,d]$ has as input a union of quantifier-free, self-join-free, and acyclic conjunctive queries $\Psi$. The goal is to decide whether the function $\mathcal{D} \mapsto \ans{\Psi}{\mathcal{D}}$ can be computed in time $O(|\mathcal{D}|^c)$, or whether it cannot be solved in time $O(|\mathcal{D}|^d)$; the behaviour may be undefined for inputs $\Psi$ for which the best exponent in the running time is in the interval $(c,d]$.
\end{definition}

\begin{theorem}
Assume that the Triangle Conjecture and non-uniform ETH hold. Then for each positive integer~$d$,   the problem $\meta[1,d]$ is $\mathrm{NP}$-hard and, furthermore, 
    \[\meta[1,d] \notin \bigcap_{\varepsilon>0} \mathrm{DTime}(2^{\varepsilon \cdot \ell})\,.\]
    This remains true even if the input to $\meta$ is restricted to be over a binary signature.
\end{theorem}
\begin{proof}
By Lemma~\ref{lem:non_inif_ETH_clique} there is a positive integer~$t$ such that for all positive integers~$k$, the function $\calD\mapsto\homs{\mathcal{K}_t^k}{\calD}$ cannot be computed in time $O(|\mathcal{D}|^d)$.
Fix this $t$ and proceed similarly to the proof of Lemma~\ref{lem:meta_hard3}. 
\end{proof}

Corollary~\ref{cor:cor} is an immediate consequence, since any algorithm that solves $\meta[c,d]$ for $1 \leq c\leq d$ solves, without modification, $\meta[1,d]$.

\begin{corollary}\label{cor:cor}
Assume that the Triangle Conjecture and non-uniform ETH hold. Then for every pair $(c,d)$ of positive integers satisfying  $c \leq d$, the problem $\meta[c,d]$ is $\mathrm{NP}$-hard and, furthermore, 
    \[\meta[c,d] \notin \bigcap_{\varepsilon>0} \mathrm{DTime}(2^{\varepsilon \cdot \ell})\,.\]
    This remains true even if the input to $\meta$ is restricted to be over a binary signature.
\end{corollary}

\section{Connection to the Weisfeiler-Leman-Dimension}
Recall that we call a database a \emph{labelled graph} if its signature has arity at most $2$, and if it does not contain a self-loop, that is, a tuple of the form $(v,v)$. Moreover, (U)CQs on labelled graphs must also have signatures of arity at most $2$ and must not contain atoms of the form $R(v,v)$, where $R$ is any relation symbol of the signature.

Neuen~\cite{Neuen23} and Lanzinger and Barceló~\cite{LanzingerB23} determined the WL-dimension of computing finite linear combinations of homomorphism counts to be the \emph{hereditary treewidth}, defined momentarily. Since counting homomorphisms is equivalent to counting answers to quantifier-free conjunctive queries, and since the number of answers of a union of conjunctive queries can be expressed as a linear combination of conjunctive query answer counts (Lemma~\ref{lem:UCQ_to_hombasis}), we can state their results as follows.
\newcommand{\hdtw}{\mathsf{hdtw}}
\begin{definition}[Hereditary Treewidth of UCQs]
    Let $\Psi$ be a UCQ. The \emph{hereditary treewidth} of $\Psi$, denoted by $\hdtw(\Psi)$, is defined as follows:
    \[ \hdtw(\Psi)= \max\{\tw(A,X) \mid c_\Psi(A,X) \neq 0 \} ,\]
    that is, $\hdtw(\Psi)$ is the maximum treewidth of any conjunctive query that survives with a non-zero coefficient when $\Psi$ is expressed as a linear combination of conjunctive queries. 
\end{definition}

Then, applying the main result of Neuen, Lanzinger and Barceló to UCQs, we obtain:
\begin{theorem}[\cite{Neuen23,LanzingerB23}]
    Let $\Psi$ be a quantifier-free UCQ on labelled graphs. Then $\mathsf{dim}_{\mathrm{WL}}(\Psi) =\hdtw(\Psi)$.
\end{theorem}

This enables us to prove Theorem~\ref{thm:WL1}, which we restate for convenience.
\WLone*

\begin{proof}
    For the upper bound, we compute the coefficients \[ c_\Psi(\mathcal{A},X) = \sum_{\substack{J\subseteq [\ell]\\ \conjuncts{\Psi}{J} \cong (\mathcal{A},X)}} (-1)^{|J|+1},\]
    that is, for each subset $J\subseteq [\ell]$, compute $\conjuncts{\Psi}{J}$. Afterwards, collect the isomorphic terms and compute $c_\Psi(\conjuncts{\Psi}{J})$ for each $J\subseteq[\ell]$. Clearly, $c_\Psi(\mathcal{A},X)=0$ for every $(\mathcal{A},X)$ that is not isomorphic to any $\conjuncts{\Psi}{J}$. Clearly, this can be done in time $|\Psi|^{O(1)}\cdot O(2^\ell)$. 
    
    Next, for each $(\mathcal{A},X)$ with $c_\Psi(\mathcal{A},X)\neq 0$, we use the algorithm of Feige, Hajiaghayi, and Lee~\cite{FeigeHL08} to compute in polynomial time a $g$-approximation $S(\mathcal{A},X)$ of the treewidth of $(\mathcal{A},X)$, where $g(k)\in O(\sqrt{\log k})$. Finally, we output the maximum of $S(\mathcal{A},X)$ over all $(\mathcal{A},X)$ with $c_\Psi(\mathcal{A},X)\neq 0$

    For the lower bound, assume that there is a function $f:\mathbb{Z}_{>0} \to \mathbb{Z}_{>0}$ and an algorithm $\mathbb{A}$ that computes an $f$-approximation of $\mathsf{dim}_{\mathrm{WL}}(\Psi)$ in subexponential time in the number of conjunctive queries in the union. We will use $\mathbb{A}$ to construct a subexponential time algorithm for $3$-$\textsc{SAT}$, which refutes ETH. Our construction is similar to the proof of Lemma~\ref{lem:meta_hard1}. Fix any positive integer $t> f(1)+1$.

    Let $F$ be a $3$-CNF with $n$ variables, which we can again assume to be sparse by using the Sparsification Lemma~\cite{CalabroIP09} (the details are identical to its application in the proof of Lemma~\ref{lem:meta_hard1}). Next, using~\cite{RouneS13}, we obtain a complex $\Delta$, the reduced Euler characteristic of which is zero if and only if $F$ is not satisfiable. Finally, we apply Lemma~\ref{lem:main_reduct_meta} with our choice of $t$. The corresponding algorithm computes in polynomial time either the reduced Euler characteristic of $\Delta$, or otherwise outputs a UCQ $\Psi$ such that the number $\ell$ of CQs in the union is bounded by $O(n)$. Moreover, the hereditary treewidth of $\Psi$ is $1$ if $\hat{\chi}(\Delta)= 0$, i.e., if $F$ is not satisfiable; and its hereditary treewidth is at least $\tw(\mathcal{K}^k_t)=t-1 > f(1)$, otherwise.

    Thus, we run $\mathbb{A}$ on $\Psi$ and report that $F$ is satisfiable if and only if it outputs $S>f(1)$. Since $\mathbb{A}$ runs in time $2^{o(\ell)}$, the total running time is bounded by $2^{o(n)}$, which refutes ETH. $\mathrm{NP}$-hardness follows likewise.
\end{proof}

Finally, the proof of Theorem~\ref{thm:WL2} is identical with the only exception being that, since $k$ is fixed, we can substitute the approximation algorithm for treewidth of Feige, Hajiaghayi, and Lee~\cite{FeigeHL08} by the exact algorithm of Bodlaender~\cite{Bodlaender96}, which runs in polynomial time if $k$ is fixed.

\section{Acknowledgements}

We thank the referees for very useful comments.

\bibliographystyle{plain}
\bibliography{ucq}

\appendix

\section{Necessity of the Side Conditions in Theorem~\ref{thm:main_dicho_quantifiers}}\label{sec:appendix}
We show that Theorem~\ref{thm:main_dicho_quantifiers} is optimal in the sense that, if any of the conditions (I), (II), or (III) is dropped, the statement of the theorem becomes false (assuming that W[1]-hard problems are not fixed-parameter tractable).

\paragraph*{Dropping condition (I)}

\begin{lemma}
There is a recursively enumerable class $C$ of quantifier-free 
UCQs of bounded arity  such that $\conjunctclass{C}$ has unbounded treewidth but $\#\ucq(C)$ is fixed-parameter tractable.
The class $C$ satisfies (II) and (III).
\end{lemma}
\begin{proof}
Let $\Delta$ be the second complex in Figure~\ref{fig:Example_complexes}, that is, the ground set is $\Omega=\{1,2,3,4\}$ and the facets are $\{1,2\}, \{2,3\},\{3,1\}$, and $\{4\}$. Note that $\Delta$ is irreducible (no element dominates another element), non-trivial, and $\Omega$ is not a facet. We can thus use Lemma~\ref{lem:main_reduct_meta_help} and let $\Psi_t$ to be the output of algorithm $\hat{\mathbb{A}}_t$ given $\Delta$. Note that $\Psi_t$ is quantifier-free --- in particular, this implies that all conjunctive queries within $\Psi_t$ are \cminimal. Since the algorithm $\hat{\mathbb{A}}_t$ can be explicitly constructed from~$t$ (see Lemma~\ref{lem:main_reduct_meta_help} and Algorithm~\ref{alg:cap}) 
the class $C=\{\Psi_t \mid t\geq 1\}$ is recursively enumerable. 
Furthermore, all of the relation symbols in queries in UCQs in~$C$ have arity~$2$, so $C$ has bounded arity.

By Item~1 of Lemma~\ref{lem:main_reduct_meta_help}, $\conjunctclass{C}$ has unbounded treewidth, since the treewidth of $\mathcal{K}_t^k$ is equal to $t-1$. Moreover, by Item~5 of~Lemma~\ref{lem:main_reduct_meta_help}, all $\Psi_t$ are unions of self-join-free conjunctive queries. We next show that $\#\ucq(C)$ is fixed-parameter tractable.

Recall that $\hat{\chi}(\Delta)= -(3 - 4 +1)= 0$.
Item~2 of~Lemma~\ref{lem:main_reduct_meta_help}
shows that $c_{\Psi_t}(\wedge(\Psi_t)) = 0$. 
Item~3 shows that for any relational structure $\mathcal{B}$ that
is not isomorphic to $\wedge(\Psi_t)$ 
with $c_{\Psi_t}(\mathcal{B})\neq 0$, $\mathcal{B}$ is acyclic. 
Recall that $\Gamma(C)$ is the class of those conjunctive queries that contribute to the CQ expansion of at least one UCQ in $C$.
So all CQs in $\Gamma(C)$ are acyclic, which means that  the treewidth of $\Gamma(C)$ is bounded by $1$.  

Since each query in $\Gamma(C)$ is quantifier-free, and is thus its own contract,  $\Gamma(C)=\contract{\Gamma(C)}$. Thus, by Theorem~\ref{thm:implicit_UCQ_dicho}, $\#\ucq(C)$ is fixed-parameter tractable.

We finish the proof by showing that $C$ satisfies (II) and (III).
Item (II) - that the number of existentially qunatified variables of queries in~$C$ is bounded - is trivial, because there are none.
We have already noted (III) -- that the UCQs in $C$ are unions of self-join-free CQs.

\end{proof}

\paragraph*{Dropping condition (II)}

\begin{lemma}
There is a recursively enumerable and deletion-closed class $C$ of unions of self-join-free conjunctive queries of bounded arity such that $\conjunctclass{C}$ has unbounded treewidth but $\#\ucq(C)$ is fixed-parameter tractable. 
\end{lemma}
\begin{proof}

The statement of the lemma guarantees that $C$ satisfies items (I) and (III) of Theorem~\ref{thm:main_dicho_quantifiers}.
     
Let $k\geq 3$ be a positive integer and let $\tau_k=(E_1,\dots,E_k)$ be a signature with $\mathsf{arity}(E_i)=2$ for all $i\in[k]$.
For any pair $i,j\in [k]$ with $i< j$, consider the conjunctive query
\[\varphi_k^{i,j}(x_1,\dots,x_k,x_\bot) = \exists y_k^{i,j} : E_i(x_i,y_k^{i,j}) \wedge E_j(x_j,y_k^{i,j}) \wedge  \bigwedge_{\ell \in [k]\setminus\{i,j\}} E_\ell(x_\ell,x_\bot)   . \]
Let $\Psi_k = \bigvee_{i< j \in[k]} \varphi_k^{i,j}$ and let $C$ be obtained from the class $\{\Psi_k \mid k\geq 3 \}$ by taking the closure under deletion of conjunctive queries. Clearly, $C$ is recursively enumerable.

Note that each conjunctive query $\varphi_k^{i,j}$ is self-join-free and that $\Psi_k$ contains $\binom{k}{2}$ existentially quantified variables.   Thus the number of existentially quantified variables of queries in $C$ is unbounded.   Moreover,   the treewidth of $\conjunctclass{C}$ is unbounded. To see this, observe that 
\[\conjunctfull{\Psi_k}(x_1,\dots,x_k) = \exists y_k^{1,2},y_k^{1,3},\dots,y_k^{k-1,k} \colon \bigwedge_{i< j} E_i(x_i,y_k^{i,j}) \wedge E_j(x_j,y_k^{i,j}) \wedge\bigwedge_{\ell \in [k]}  E_\ell(x_\ell,x_{\bot}) . \]
Therefore, the Gaifman graph of $\conjunctfull{\Psi_k}$ contains as a subgraph a subdivision of a $k$-clique and  thus has treewidth at least $k-1$.

It remains to show that $\#\ucq(C)$ is fixed-parameter tractable. To show this  we claim that the classes $\Gamma(C)$ and $\contract{\Gamma(C)}$ have treewidth at most $2$, and thus the problem $\#\ucq(C)$ is fixed-parameter tractable by Theorem~\ref{thm:implicit_UCQ_dicho}.  To  prove the claim, fix any $k\geq 3$ and any non-emtpy subset $J\subseteq\{(i,j) \in[k^2] \mid i< j\}$.  We will show that the \cminimal~representatives of $\conjunctfull{\Psi_k,J}$ and its contract are acyclic. To this end, assume first that $|J|=1$. Then $\conjunctfull{\Psi_k,J}$ is equal to one of the conjunctive queries $\varphi_k^{i,j}$, which is clearly acyclic.  
Since $\varphi_k^{i,j}$ is self-join-free and does not contain isolated variables, it  is \cminimal~by Lemma~\ref{lem:self-join-free-counting-cores}. Let $G$ be the Gaifman graph of $\varphi_k^{i,j}$, and recall that the contract of $\varphi_k^{i,j}$ is obtained from $G[X]$ by adding an edge between two free variables in $X$ if and only if there is a connected component in the quantified variables that is adjacent to both free variables. Since the only quantified variable in $\varphi_k^{i,j}$ is $y_k^{i,j}$, which is adjacent (in $G$) to $x_i$ and $x_j$, the contract of $\varphi_k^{i,j}$ is just the graph obtained from $G[X]$ by adding an edge between $x_i$ and $x_j$, which also yields an acyclic graph.
    
Next assume that $|J|\geq 2$. For an index $s \in[k]$, we say that $J$ \emph{covers} $s$ if each $(i,j)\in J$ satisfies $i=s$ or $j=s$.
We distinguish three cases:
\begin{itemize}
\item[(A)] There are distinct $s_1 < s_2$ such that $J$ covers $s_1$ and $s_2$. Then $J=\{(s_1,s_2)\}$, contradicting the assumption that $|J|\geq 2$.
\item[(B)] There is precisely one $s \in[k]$ such that $J$ covers $s$. Assume w.l.o.g.\ that $s=k$. Then 
since every $(i,j) \in J$ has $i<j$
\begin{align*} \conjunctfull{\Psi_k,J} &= \bigwedge_{\ell\in[k-1]} E_\ell(x_\ell,x_\bot) \wedge \bigwedge_{(i,j)\in J} \exists y_k^{i,j} : E_i(x_i,y_k^{i,j}) \wedge E_j(x_j,y_k^{i,j})  
\\
&= \bigwedge_{\ell\in[k-1]} E_\ell(x_\ell,x_\bot) \wedge \bigwedge_{(i,k)\in J} \exists y_k^{i,k} : E_i(x_i,y_k^{i,k}) \wedge E_k(x_k,y_k^{i,k})  .\end{align*} 
 
Observe that any answer of $\conjunctfull{\Psi_k,J}$ in a database $\mathcal{D}$ is also an answer of the following query, and vice versa:
\[\psi_k := \bigwedge_{\ell\in[k-1]} E_\ell(x_\ell,x_\bot)  \wedge \bigwedge_{i\in[k-1]} \exists y_k^{i,k} : E_i(x_i,y_k^{i,k}) \wedge E_k(x_k,y_k^{i,k}) 
,\]
since all $y_k^{i,\ell}$ with $i< k$ can be mapped to the same vertex as $x_\bot$.
Thus $\conjunctfull{\Psi_k,J}$ and $\psi_k$ are counting equivalent. Moreover, $\psi_k$ is self-join-free and does not contain isolated variables. Thus, by Lemma~\ref{lem:self-join-free-counting-cores}, it is \cminimal. Finally,  deleting $x_k$ from the Gaifman graph of $\psi_k$ yields an acyclic graph, and the same is true for the contract of $\psi_k$. Therefore, the treewidth of both $\psi_k$ and its contract are at most $2$.
\item[(C)] There is no $s\in[k]$ such that $J$ covers $s$. Then 
\[\conjunctfull{\Psi_k,J} = \bigwedge_{\ell\in[k]} E_\ell(x_\ell,x_\bot) \wedge \bigwedge_{(i,j)\in J} \exists y_k^{i,j} : E_i(x_i,y_k^{i,j}) \wedge E_j(x_j,y_k^{i,j})  .\]
Observe that any answer of $\conjunctfull{\Psi_k,J}$ in a database $\mathcal{D}$ is also an answer of the following query, and vice versa:
\[\psi_k := \bigwedge_{\ell\in[k]} E(x_\ell,x_\bot) \,,\]
since all $y_k^{i,j}$ can be mapped to the same vertex as $x_\bot$. 

Thus $\conjunctfull{\Psi_k,J}$ and $\psi_k$ are counting equivalent. Since $\psi_k$ does not contain quantified variables, it must be both its own \ccore~and its own contract. This concludes the proof of the claim since $\psi_k$ is acyclic.
\end{itemize}
\end{proof}

\paragraph*{Dropping condition (III)}

\begin{lemma}
There is a recursively enumerable and deletion-closed class $C$ of quantifier-free 
UCQs of bounded arity  such that $\conjunctclass{C}$ has unbounded treewidth but $\#\ucq(C)$ is fixed-parameter tractable.
\end{lemma}
\begin{proof}

The statement of the lemma guarantees that $C$ satisfies items (I) and (II) of Theorem~\ref{thm:main_dicho_quantifiers}.
We show an even stronger claim by requiring $C$ to be a recursively enumerable class of quantifier-free CQs (instead of UCQs) of bounded 
arity such that   $\conjunctclass{C}$ has unbounded treewidth but $\#\ucq(C)$ is polynomial-time solvable. Note that each conjunctive query is a (trivial) union of conjunctive queries; moreover, this also means that $C$ is deletion-closed. 
    
    For each $k\geq 1$, we define a conjunctive query $\psi_k$ over the signature of graphs as follows:
    \[ \psi_k(x_1,\dots,x_k,x_\bot) = \exists y: \bigwedge_{i\in[k]} E(x_i,x_\bot) \wedge E(x_i,y)  \,.\]
    The query $\psi_k$ has only one quantified variable. Moreover, the contract of $\psi_k$ is a $k$-clique and thus has treewidth $k-1$. However, $\psi_k$ is clearly \cequivalent~to the query 
    \[ \psi'_k = \bigwedge_{i\in[k]} E(x_i,x_\bot) \,, \]
    which is its own contract (since there are no quantified variables), and which is of treewidth $1$. Thus, for $C$ being the class of all $\psi_k$, we find that $\contract{\conjunctclass{C}}$ has unbounded treewidth, but, according to Theorem~\ref{thm:cq_classification}, the problem $\#\ucq(C)$ is solvable in polynomial time.
\end{proof}

\end{document}